\documentclass[11pt]{article}

\usepackage{amssymb,amsmath,amsfonts}
\usepackage{graphicx,xcolor,enumitem}
\usepackage{epsfig}
\usepackage{amsthm}
\usepackage{bm}
\usepackage{subcaption}
\usepackage[round]{natbib}
\usepackage{csquotes}
\renewcommand{\mkbegdispquote}[2]{\itshape}
\usepackage[a4paper, twoside, hmarginratio=1:1, vmarginratio=1:1, left=1in,top=1in]{geometry}

\RequirePackage[breaklinks=true, hidelinks]{hyperref}
\usepackage{breakcites}

\newcommand{\rdot}{{r \wedge \cdot}}
\newcommand{\bd}{\mathbf{d}}
\newcommand{\bu}{\mathbf{u}}

\newcommand{\bA}{\mathbf{A}}
\newcommand{\bpi}{\bm{\pi}}
\newcommand{\p}{\mathbb{P}}
\newcommand{\F}{\mathbb{F}}
\newcommand{\E}{\mathbb{E}}
\newcommand{\R}{\mathbb{R}}
\newcommand{\cF}{{\mathcal F}}
\newcommand{\cU}{{\mathcal U}}
\newcommand{\id}{\mathbf{1}}
\newcommand{\as}{\mbox{{\rm a.s.}}}

\newcommand{\hatE}{ \hat{\mathbb{E}} }
\newcommand{\hE}{ \mathbb{E}^h }

\newcommand{\ang}[1]{\langle  #1 \rangle } 

\newtheorem{theorem}{Theorem}

\newtheorem{assumption}[theorem]{Assumption}

\newtheorem{corollary}[theorem]{Corollary}

\newtheorem{definition}[theorem]{Definition}

\newtheorem{lemma}[theorem]{Lemma}

\newtheorem{remark}[theorem]{Remark}

\theoremstyle{definition}

\numberwithin{equation}{section}
\numberwithin{theorem}{section}

\begin{document}

\title{Time-Inconsistency with Rough Volatility\footnote{An earlier version of this paper was circulated and cited under the title ``Time-consistent feedback strategies with Volterra processes". The authors would like to thank two anonymous referees and the editors for their careful reading and valuable comments, which have greatly improved the manuscript. Bingyan Han is supported by UIC Start-up Research Fund (Reference No: R72021109).}}
\author{Bingyan Han \thanks{Division of Science and Technology, BNU-HKBU United International College, Zhuhai, Guangdong, China, bingyanhan@uic.edu.cn}
	\and Hoi Ying Wong\thanks{Department of Statistics, The Chinese University of Hong Kong, Hong Kong SAR, China, hywong@cuhk.edu.hk}
}

\date{December 21, 2021}
\maketitle

\maketitle

\begin{abstract}
	 In this paper, we consider equilibrium strategies under Volterra processes and time-inconsistent preferences embracing mean-variance portfolio selection (MVP). Using a functional It\^o calculus approach, we overcome the non-Markovian and non-semimartingale difficulty in Volterra processes. The equilibrium strategy is then characterized by an extended path-dependent Hamilton-Jacobi-Bellman equation system under a game-theoretic framework. A verification theorem is provided. We derive explicit solutions to three problems, including MVP with constant risk aversion, MVP for log returns, and a mean-variance objective with a linear controlled Volterra process. We also thoroughly examine the effect of volatility roughness on equilibrium strategies. Numerical experiments demonstrate that trading rules with rough volatility outperform the classic counterparts.
	\\[2ex] 
	\\[2ex]
\end{abstract}

\section{Introduction}
Portfolio selection is a fundamental and leading problem in mathematical finance. Pioneered by \cite{markowitz1952portfolio}, mean-variance portfolio selection (MVP) is well recognized as a cornerstone of modern portfolio theory. Its intuitive and flexible formulation has attracted the attention of numerous researchers, who have sought to strengthen the original framework. Examples include, but are not limited to, \cite{li2000optimal,zhou2000continuous,basak2010dynamic,czichowsky2013time,bjork2014mean,he2018equilibrium,dai2020log}. In contrast to expected utility theory \citep{merton1969lifetime}, MVP suffers from time-inconsistency induced by the variance operator. When an investor with an initial value of $(t, x)$ perceives that the derived strategy is no longer optimal at a later state $(s, X_s)$ for $s > t$, time-inconsistency occurs. As noted in \cite{strotz1955myopia}, time-consistency is a fundamental requirement for any reasonable strategic decision-making and optimization. An agent should only choose strategies from which he/she will not deviate. \cite{basak2010dynamic} identifies that MVP has an adjustment term that provides ``an incentive for the investor to deviate from his optimal strategy at a later time."  However, studies such as \cite{li2000optimal,zhou2000continuous} have neglected time-inconsistency and have only provided the pre-committed strategy. 

Time-inconsistency has generated an astonishing amount of controversial opinions toward the notion of {\it optimality}. A remedy is to consider an equilibrium strategy. The intent is to formulate a game between the current agent and his/her future selves and to then derive an equilibrium of the game as the strategy. Several treatments are available based on different methodologies. \cite{ekeland2008investment,bjork2014mean,bjork2017time} follow the classic dynamic programming framework and derive an extended Hamilton-Jacobi-Bellman (HJB) equation to characterize the equilibrium. \cite{hu2012time} formulates the game in an open-loop optimization and derives the equilibrium via the stochastic maximum principle. \cite{yong2012time,czichowsky2013time} consider a partition on the whole planning time horizon and obtain the equilibrium by taking limits. With a fixed-point argument, \cite{huang2018strong} further distinguishes between strong and weak equilibria. Related discussions include \cite{he2018equilibrium} and the references therein. In this paper, we exploit the extended HJB methodology for its wider applications, including MVP. 

\cite{liu2007portfolio,basak2010dynamic,dai2020log} document that hedging demand from stochastic volatility can comprise a substantial percentage of total equilibrium stock exposure. Realistic modeling for volatility then plays an indispensable role in investment decision-making. We use rough volatility models, recently proposed by \cite{gatheral2018volatility}. In this seminal work, \cite{gatheral2018volatility} uses the fractional Brownian motion (fBm) to demonstrate the roughness of volatility. By developing a rigorous statistical estimation and inference, \cite{fukasawa2019volatility} further confirms that volatility is even rougher than reported in \cite{gatheral2018volatility}. These models are consistent with some stylized facts of financial time series and have several desired theoretical properties. They capture the term structure of the implied volatility (IV) surface, especially for the explosion of the at-the-money (ATM) skew when the maturity is close to zero \citep{alos2007short,gatheral2018volatility,el2019characteristic}, which smooth volatility models fail to do. Examples of rough volatility models include the fBm \citep{gatheral2018volatility}, the fractional Ornstein-Uhlenbeck (fOU) process \citep{fouque2019optimalMF,fouque2018optimalSIFN}, the rough Bergomi (rBergomi) model \citep{bayer2016pricing}, the fractional Heston model \citep{guennoun2018asymptotic}, and the rough Heston model \citep{el2019characteristic}. The rough Heston model has received particular attention and has been extended to the Volterra Heston model \citep{abi2017affine} and the affine forward variance (AFV) model \citep{gatheral2019affine}. The economic interpretation of rough volatility is explained by \cite{el2019characteristic} via metaorders in high-frequency trading, by \cite{jusselin2018no} via the no-arbitrage property, by \cite{glasserman2019buy} via heterogeneity in near-term downside risk, and by \cite{gatheral2018volatility} via long memory behavior. However, like the debate on the short-range or long-range dependence in volatility, the understanding of rough volatility is still in development.

We are interested in the equilibrium strategies under a more realistic stochastic financial environment delineated by rough volatility. We are the first to explore time-inconsistency with rough volatility, although related works, such as \cite{fouque2019optimalMF,han2019mean,han2019merton,fouque2018optimalSIFN,bauerle2020portfolio}, and the references therein are available for alternative portfolio problems under rough volatility. Despite empirical evidence of rough volatility, its non-Markovian and non-semimartingale nature challenges the classic methodology of equilibria. We also consider Volterra processes for generality. Previous results in the literature cannot be directly adopted to tackle equilibrium strategies under Volterra processes. \cite{hu2012time} accounts for the linear non-Markovian systems, but the application is limited to linear-quadratic control problems. The locally mean-variance efficient (LMVE) approach in \cite{czichowsky2013time} can deal with semimartingales but the Volterra process is not typically a semimartingale.

To tackle these difficulties, we adopt a general methodology called functional It\^o calculus with applications far beyond time-inconsistency and rough volatility. \cite{dupire2019functional} first develops a pathwise calculus for non-anticipative functionals, motivated by pricing and hedging path-dependent derivatives. By defining the time and spatial derivatives, the classic It\^o formula is extended to the functional It\^o formula for path-dependent functionals in \cite{cont2013functional,dupire2019functional}. The functional It\^o calculus is useful for a wide class of optimization and decision-making problems. To better incorporate financial and insurance risks, \cite{siu2016functional} considers the applications of the functional It\^o calculus in convex risk measures with non-Markovian jump-diffusion processes. \cite{cvitanic2017moral} uses Dupire's functional It\^o calculus to motivate their definition of contracts in principal-agent problems. Under the framework of \cite{cont2013functional,dupire2019functional}, \cite{schied2018model} presents two pathwise versions of the master formula in Fernholz's stochastic portfolio theory and elucidates their performance with empirical data. \cite{saporito2019stochastic} applies the functional It\^o calculus to stochastic differential games and optimal control problems with delay. 

However, the aforementioned literature relies on a semimartingale assumption, whereas the Volterra Heston model and general Volterra processes are non-semimartingales. Recently, \cite{viens2017martingale} has developed a powerful toolkit for the functional It\^o formula to analyze functionals of Volterra processes. Heuristically, their approach aims to ``recover" the flow property of Volterra processes by incorporating an auxiliary non-anticipative process $\Theta^t$ (\ref{Eq:Theta}) into the path $\omega$. Their elegant results enable us to derive the extended path-dependent HJB (PHJB) equation system in Theorem \ref{Thm:Verification}, which extends the results of \cite{zhao2014consumption,bjork2014mean,bjork2017time,dai2020log} and has potential applications in addition to portfolio selection under rough volatility. Nevertheless, we stress that the development of the PHJB system is non-trivial even given the existing results. We also offer an example of the unsolved future problem suggested in \cite{bjork2017time}:
\begin{displaycquote}{bjork2017time}
	``The present theory depends critically on the Markovian structure. It would be interesting to see what can be done without this assumption.''
\end{displaycquote}

In Section \ref{Sec:Example}, we apply the general framework for Volterra processes to time-consistent (TC) MVP under rough volatility. We refer to the agent under a rough stochastic environment as a rough investor. Should a rough investor buy more or less when volatility is rougher? When should he/she change his/her preference to rough? Furthermore, how profitable are rough strategies according to empirical data? We observe that volatility roughness dramatically increases the demand for hedging. Additionally, rough investors' attitude toward roughness depends on the payoff functional in mind. By deriving explicit solutions to classic problems in \cite{basak2010dynamic,dai2020log}, we present a thorough analysis and disentangle the complicated relationship between time-inconsistency and roughness. Our findings advocate rough volatility models as a promising alternative for the classic models adopted in \cite{basak2010dynamic,dai2020log}.
\begin{itemize}
	\item Section \ref{Sec:ConstMV} considers TC-MVP with constant risk aversion \citep{basak2010dynamic} under the Volterra Heston model. It is referred to as the const-MV case. In our sensitivity analysis, we find that when the investment horizon is long, the const-MV strategy suggests high demand on stocks with smoother volatility. For a short investment horizon, the const-MV strategy increases exposure to stocks with rougher volatility. We refer to this phenomenon as the {\it investment horizon effect}. In the const-MV case, the change point in this effect is irrelevant to the investor's risk aversion. In a simulation study, we find that a rough investor demands up to 40\% more than a smooth investor.
	\item Section \ref{Sec:LogMV} investigates TC-MVP for log returns \citep{dai2020log} under a rough stochastic environment. It is referred to as the log-MV case. The investment horizon effect in the const-MV case is retained in the log-MV case. However, the main difference is that heterogeneity in risk aversion changes when investors start to prefer rough. A more risk-averse investor prefers rough much earlier. In general, the log-MV case prohibits bankruptcy and is thus more conservative than the const-MV case. Rough investors increase their stock demand by at most 9\% compared with \cite{dai2020log}, as detailed in the simulation study. In addition, the log-MV criterion implies different roughness-related behaviors, compared with constant relative risk aversion (CRRA) utility.
	\item Section \ref{Sec:LC-MV} studies a problem with the mean-variance (MV) objective and a linear controlled Volterra process. Distinct from portfolio selection, the convolution and the control appear together in the state process \eqref{Eq:LC-MVstate}. We manage to derive an explicit strategy, which demonstrates the potential application of our framework to problems with controlled Volterra state processes.
	\item Via the empirical study in Section \ref{Sec:Empirical}, we highlight that trading strategies with rough volatility dominate all classic counterparts in \cite{basak2010dynamic,dai2020log}. They gain greater terminal wealth with a better Sharpe ratio, even under the volatile U.S. equity markets from May 2018 to April 2019. The empirical performance substantiates the claim that our proposed strategies are more effective in capturing market movements.
\end{itemize}

The reminder of this paper is organized as follows. Section \ref{Sec:Formulation} describes the general framework with Volterra processes. Section \ref{Sec:VolHeston} reviews the Volterra Heston model as a main example of rough volatility. Section \ref{Sec:EP-HJB} derives the extended PHJB equation system. Section \ref{Sec:Example} discusses the solutions to three problems using the MV objective. Section \ref{Sec:Num} presents the numerical study. Finally, Section \ref{Sec:Con} concludes the paper. The functional It\^o calculus in \cite{viens2017martingale} is summarized in Appendix \ref{Sec:FuncIto} for a self-contained article. All mathematical proofs are deferred to Appendix \ref{Sec:Proofs}.

\section{Problem formulation}\label{Sec:Formulation}
\subsection{Volterra processes}
Let $T>0$ be a deterministic finite investment horizon. We first present a general model with applications beyond rough volatility. Consider a controlled $n$-dimensional stochastic Volterra integral equation (SVIE) on $[0, T]$:
\begin{equation}\label{Eq:Volterra0}
X^\bu_t = x + \int^t_0 \mu(t; r, X^\bu_\rdot, \bu(r, X^\bu_\rdot)) dr + \int^t_0 \sigma(t; r, X^\bu_\rdot, \bu(r, X^\bu_\rdot)) dW_r,
\end{equation}
where $X^\bu_\rdot$ refers to the whole past path of the process $(X^\bu_s)_{0 \leq s \leq r}$, $W$ is a $d$-dimensional standard Brownian motion, and $\mu$, $\sigma$ are adapted with suitable dimensions. The feedback strategy $\bu$ is a $k$-dimensional deterministic measurable function. We provide a rigorous definition of admissible strategies in Definition \ref{Def:U}. The main example of (\ref{Eq:Volterra0}) in this paper is the two-dimensional process $(M^\bu, \nu)$ in (\ref{Eq:wealth}) and (\ref{Eq:VoltHeston}).  It is also worth mentioning again that SVIE (\ref{Eq:Volterra0}) is non-Markovian and non-semimartingale in general.

We consider feedback strategies $\bu(r, X^\bu_\rdot)$ that depend on the whole path $X^\bu_\rdot$ instead of solely on the current value $X^\bu_r$ of the process as in \cite{basak2010dynamic,bjork2017time,dai2020log}. This setting is more reasonable because investors can always base their decisions on the observed history of the process. 

Before formally defining the equilibrium feedback strategies, we impose the following standing assumption throughout this paper.
\begin{assumption}\label{Assum:SVIE}
	Controlled SVIE (\ref{Eq:Volterra0}) admits a unique in law continuous weak solution $(X^\bu, W)$ on some complete probability space $(\Omega, \cF, \p)$ with a filtration $\F = \{ \cF_t \}_{ 0 \leq t \leq T}$ that satisfies the usual conditions. $(X^\bu, W)$ is $\F$-adapted and
	\begin{equation}
	\E \Big[ \sup_{ 0 \leq t \leq T} |X^\bu_t |^p \Big] < \infty,
	\end{equation}
	for any $ p \geq 1$.
\end{assumption}

As noted in the Appendix of \cite{viens2017martingale}, a sufficiently large moment constant $p$ is enough. However, this $p$ is not explicit. As our primary focus is time-inconsistency, we do not pursue potentially more general conditions validating Assumption \ref{Assum:SVIE} in this paper. We verify Assumption \ref{Assum:SVIE} for some examples in Section \ref{Sec:Example} and refer interested readers to \cite{abi2017affine,viens2017martingale} for further results. Indeed, as $\bu$ is a feedback strategy, we can regard $\mu^\bu$ and $\sigma^\bu$ in \eqref{Eq:muu} as the drift and the diffusion, respectively. The results without controls in \cite{abi2017affine,viens2017martingale} then become applicable to our cases. Compared with Assumption 3.1 of \cite{viens2017martingale}, Assumption \ref{Assum:SVIE} further requires (\ref{Eq:Volterra0}) to admit a unique in law solution. For a given feedback strategy $\bu$, it is natural for our problem to attain a unique reward functional (\ref{Eq:Reward}) under $\bu$, which requires the law of SVIE (\ref{Eq:Volterra0}) to be unique. We also need a continuous solution to (\ref{Eq:Volterra0}). This condition is relatively mild when a feedback strategy is considered. The concatenated path (\ref{Eq:omega}) is later justified as continuous under this condition. We fix a weak solution $(X^\bu, W)$ to (\ref{Eq:Volterra0}) once the feedback strategy $\bu$ is given. Although the probability space and Brownian motion are also parts of the solution to SVIE \eqref{Eq:Volterra0}, the distribution of $X^\bu$ is unique. We denote $\E[\cdot]$ and $\E[\cdot|\cF_t]$ as the expectation and conditional expectation, respectively, under the probability measure $\p$, which is a part of the weak solution under a generic control $\bu$. As in \cite{viens2017martingale}, we use $\F = \F^{X^\bu, W}$.

\subsection{Example: Volterra Heston model}\label{Sec:VolHeston}
Consider a two-dimensional standard Brownian motion $W \triangleq (W_1, W_2)$. An important example of (\ref{Eq:Volterra0}) is a rough version of the classic Heston model \cite{el2019characteristic}:
\begin{align}\label{Eq:RoughHeston}
\nu_{t}= \nu_0 & + \frac{1}{\Gamma(H + \frac{1}{2})}  \int_0^t (t - r)^{H - 1/2} \kappa (\phi - \nu_r ) dr  + \frac{1}{\Gamma(H + \frac{1}{2})} \int_0^t (t - r)^{H-1/2} \sigma \sqrt{\nu_r} d B_r, 
\end{align}
where $\Gamma(\cdot)$ is the Gamma function and $H$ is the Hurst parameter. $dB_r = \rho dW_{1r} + \sqrt{1 - \rho^2} dW_{2r}$ and $\nu_0, \kappa, \phi$, and $\sigma$ are positive constants. The correlation $\rho$ between stock price and volatility is also constant. When $H = 1/2$, the model is reduced to the classic Heston model adopted in \cite{basak2010dynamic,dai2020log}. The volatility trajectories of (\ref{Eq:RoughHeston}) have almost surely H\"older regularity $ H - \varepsilon$, for all $\varepsilon > 0$, as shown in \cite{el2019characteristic}. Therefore, (\ref{Eq:RoughHeston}) is called the rough Heston model and the Hurst parameter $H$ is an index of volatility roughness. The smaller $H$ is, the rougher the volatility is. With $H$ of order $0.1$, \citet[Section 5.2]{el2019characteristic} shows that the rough Heston model provides remarkable fits for volatility skews, including cases of extreme short maturity. Thus, it captures the near-term risks implied by ATM skew explosion.

Extending the rough Heston model (\ref{Eq:RoughHeston}), the Volterra Heston model in \cite{abi2017affine} reads as follows:
\begin{equation}\label{Eq:VoltHeston}
\nu_{t}= \nu_0 + \kappa \int_0^t K(t - r)\left(\phi - \nu_r \right) dr + \int_0^t K(t - r) \sigma \sqrt{\nu_r} d B_r,
\end{equation}
where $K(\cdot)$ is the kernel function. By setting $K(t) = \frac{t^{H-1/2}}{\Gamma(H+1/2)}$, namely the fractional kernel, (\ref{Eq:VoltHeston}) recovers (\ref{Eq:RoughHeston}). In line with \cite{abi2017affine}, we impose the following assumption on the kernel function.

\begin{assumption}\label{Assum:kernel}
	The kernel $K$ is strictly positive and completely monotone. There exists $\tau \in(0,2]$ such that $\int_{0}^{h} K(t)^{2} d t=O\left(h^\tau \right)$ and $\int_{0}^{T}(K(t+h)-K(t))^{2} d t=O(h^\tau)$ for every $T<\infty$. 
\end{assumption}

Like \cite{abi2017affine,basak2010dynamic,dai2020log}, the risky asset (stock) price $S_t$ is postulated as follows:
\begin{equation}\label{Eq:stock}
dS_t = S_t (\varUpsilon_t + \theta \nu_t) dt + S_t \sqrt{\nu_t} dW_{1t}, \quad S_0 > 0,
\end{equation}
with a deterministic bounded risk-free rate $\varUpsilon_t>0$ and constant $\theta \neq 0$. The market price of risk, or risk premium, is then given by $\theta \sqrt{\nu_t}$. Such a risk premium specification is widely used in the literature, such as in \citet[Section 2.2]{basak2010dynamic} and \cite{dai2020log,bauerle2020portfolio}. The risk-free rate $\varUpsilon_t >0$ is the return of a risk-free asset available in the market. Indeed, a general Heston specification \citep{liu2007portfolio,basak2010dynamic,dai2020log} is also tractable, as indicated in Remark \ref{Rem:GenHeston}. We adopt (\ref{Eq:stock}) to simplify the presentation.

We quote the following result from \cite{abi2017affine}, which guarantees the existence and weak uniqueness of the Volterra Heston model.
\begin{theorem}[{\citet[Theorem 7.1]{abi2017affine}}]
	Under Assumption \ref{Assum:kernel}, the stochastic Volterra equation (\ref{Eq:VoltHeston})-(\ref{Eq:stock}) has a unique in law $\R_{\geq 0} \times \R_{\geq 0}$-valued continuous weak solution for any initial condition $(S_0, \nu_0) \in \R_{\geq 0} \times \R_{\geq 0}$.
\end{theorem}

\begin{remark}
	Pathwise uniqueness for (\ref{Eq:VoltHeston})-(\ref{Eq:stock}) remains an open problem in general. We mention \citet[Proposition B.3]{abi2019multifactor} as a related result with kernel $K \in C^1([0, T], \R)$ and \citet[Proposition 8.1]{mytnik2015uniqueness} for certain smooth kernels. However, the strong uniqueness of (\ref{Eq:VoltHeston})-(\ref{Eq:stock}) is left open for singular kernels. For weak solutions, Brownian motion is free to construct as needed. In the sequel, we fix a solution $(S, \nu, W_1, W_2)$ to (\ref{Eq:VoltHeston})-(\ref{Eq:stock}) as other solutions share the same law. Furthermore, the boundary point $0$ may be reachable for the Volterra Heston model and the property of the boundary point $0$ is left open.
\end{remark}

Multiply the dollar amount of wealth in the stock by $\sqrt{\nu_t}$ and denote it as the investment strategy $\bu$. Let $M^\bu_t$ be the wealth process. Then, $M^\bu_t$ satisfies
\begin{equation}\label{Eq:wealth}
d M^\bu_t = \big(\varUpsilon_t M^\bu_t + \theta \sqrt{\nu_t} u_t \big) dt + u_t dW_{1t}, \quad M^\bu_0 > 0.
\end{equation}
It is clear that $(M^\bu, \nu)$ in (\ref{Eq:wealth}) and (\ref{Eq:VoltHeston}) is a special case of the Volterra process (\ref{Eq:Volterra0}). We handle the general time-inconsistent problems in (\ref{Eq:Volterra0}) in a unified way, whereas the applications focus on (\ref{Eq:wealth}) and (\ref{Eq:VoltHeston}).

\subsection{Equilibrium under time-inconsistent preferences}

For time-inconsistent problems, we must consider the state process starting from time $t \in [0, T)$. For $s \geq t$, the general state process (\ref{Eq:Volterra0}) can be decomposed as follows:
\begin{align*}
X^\bu_s = x &+ \int^t_0 \mu(s; r, X^\bu_\rdot, \bu(r, X^\bu_\rdot)) dr + \int^t_0 \sigma(s; r, X^\bu_\rdot, \bu(r, X^\bu_\rdot)) dW_r \\
&+ \int^s_t \mu(s; r, X^\bu_\rdot, \bu(r, X^\bu_\rdot)) dr + \int^s_t \sigma(s; r, X^\bu_\rdot, \bu(r, X^\bu_\rdot)) dW_r.
\end{align*}
Following \cite{viens2017martingale}, we define
\begin{align}\label{Eq:Theta}
\Theta^{t, \bu}_s \triangleq x & + \int^t_0 \mu(s; r, X^\bu_\rdot, \bu(r, X^\bu_\rdot)) dr + \int^t_0 \sigma(s; r, X^\bu_\rdot, \bu(r, X^\bu_\rdot)) dW_r, \;  t \leq s \leq T.
\end{align}
$t \mapsto \Theta^{t, \bu}_s$ is a semimartingale for $0 \leq t \leq s$. Using $\Theta^{t, \bu}_s$, we concatenate a path $\omega$ as
\begin{equation}\label{Eq:omega}
\omega_s =  (X^\bu \otimes_t \Theta^{t, \bu})_s \triangleq X^\bu_s \id_{\{0 \leq s < t\}} + \Theta^{t, \bu}_s \id_{\{t \leq s \leq T\}}.
\end{equation}
Although $\omega$ is defined on $[0, T]$, it is adapted to $\cF_t$. $\omega$ is $\p$-$\as$ continuous.

An interpretation of $\Theta^{t, \bu}_s$ is that it can be written as follows:
\begin{equation}
\Theta^{t, \bu}_s = \E \Big[ X^\bu_s - \int^s_t \mu(s; r, X^\bu_\rdot, \bu(r, X^\bu_\rdot)) dr \Big| \cF_t \Big],
\end{equation}
which is related to the modified forward processes \citep{keller2018affine}. It represents the current view of process distributions in the future.

Particularly, if we consider the Volterra Heston model (\ref{Eq:VoltHeston}), then
\begin{equation}\label{Eq:VolTheta}
\Theta^t_s = \nu_0 + \kappa \int_0^t K(s - r)\left(\phi - \nu_r \right) d r + \int_0^t K(s - r) \sigma \sqrt{\nu_r} d B_r,
\end{equation}
which corresponds to the variance part of $\Theta^{t, \bu}$ in (\ref{Eq:Theta}). As $\bu$ does not appear in the variance process, we exclude it from the notation, which becomes $\Theta^t$. We further denote the concatenated path $\omega$ for the variance process as follows: 
\begin{equation}
\omega^\nu_s =  (\nu \otimes_t \Theta^t)_s \triangleq \nu_s \id_{\{0 \leq s < t\}} + \Theta^t_s \id_{\{t \leq s \leq T\}}.
\end{equation}
For $\Theta^t_s$ in (\ref{Eq:VolTheta}), \citet[Equation (5.11)]{viens2017martingale} shows that $\Theta^t_s$ can be represented by the forward variance curve $\E[\nu_s|\cF_t]$. Therefore, $\Theta^t_s$ can be roughly replicated with financial products, such as variance swaps.  

At time $t$, for a realized path $\omega$, we have the following:
\begin{align}\label{Eq:Xt}
X^{t, \omega, \bu}_s = \omega_s & + \int^s_t \mu(s; r, X^{t, \omega, \bu}_\rdot, \bu(r, X^{t, \omega, \bu}_\rdot)) dr \nonumber \\
& + \int^s_t \sigma(s; r, X^{t, \omega, \bu}_\rdot, \bu(r, X^{t, \omega, \bu}_\rdot)) dW_r, \; t \leq s \leq T, \nonumber \\
X^{t, \omega, \bu}_s = \omega_s, & \quad 0 \leq s < t,
\end{align}
where the notation $X^\bu_s$ is replaced with $X^{t, \omega, \bu}_s$ to highlight its dependence on $t$ and the path $\omega$. For $t \leq s \leq T$, $\Theta^{t, \bu}_s$ is then interpreted as follows:
\begin{equation}
\Theta^{t, \bu}_s = \omega_s = x + \int^t_0 \mu(s; r, X^{t, \omega, \bu}_\rdot, \bu(r, X^{t, \omega, \bu}_\rdot)) dr + \int^t_0 \sigma(s; r, X^{t, \omega, \bu}_\rdot, \bu(r, X^{t, \omega, \bu}_\rdot)) dW_r.
\end{equation} 

For a given feedback strategy $\bu$, let
\begin{equation}\label{Eq:muu}
\mu^\bu(t; r, \omega) \triangleq \mu(t; r, \omega_\rdot, \bu(r, \omega_\rdot)), \quad \sigma^\bu(t; r, \omega) \triangleq \sigma(t; r, \omega_\rdot, \bu(r, \omega_\rdot)).
\end{equation}
As it is enough for us to consider $\mu^\bu$ and $\sigma^\bu$ with the same singularities, we encounter two cases only. If $\lim_{r\rightarrow t} \mu^\bu(t; r, \cdot) = \infty$ and $\lim_{r\rightarrow t} \sigma^\bu(t; r, \cdot) = \infty$, it is called a {\it singular} case; otherwise, if $\lim_{r\rightarrow t} \mu^\bu(t; r, \cdot) < \infty$ and $\lim_{r\rightarrow t} \sigma^\bu(t; r, \cdot) < \infty$, it is called a {\it regular} case \cite{viens2017martingale}. 

We introduce the reward functional as follows:
\begin{align}\label{Eq:Reward}
J(t, \omega; \bu) \triangleq & \E\Big[ \int^T_t C(t, \omega_t, r,  X^{t, \omega, \bu}_\rdot, \bu(r, X^{t, \omega, \bu}_\rdot)) dr + F(t, \omega_t, X^{t, \omega, \bu}_{T \wedge \cdot}) \Big| \cF_t \Big] \cr
& + G(t, \omega_t, \E[ X^{t, \omega, \bu}_T | \cF_t]),
\end{align}
where $X^{t, \omega, \bu}$ is given by (\ref{Eq:Xt}).

Functional (\ref{Eq:Reward}) has nested MV criterion (\ref{Eq:ConstMV_Reward}) as a special case. SVIE (\ref{Eq:Volterra0}) is not time-consistent due to the absence of the flow property \citep{viens2017martingale}. However, we focus on the time-inconsistency issue from the objective function $J$, which originates from its dependence on the current time $t$, the current state $\omega_t$, and the nonlinear function $G$. The reward functional \eqref{Eq:Reward} does not satisfy the Bellman optimality principle. A strategy that maximizes \eqref{Eq:Reward} at the current time $t$ may no longer be optimal at a certain future time. We refer readers to \cite{basak2010dynamic,zhao2014consumption,bjork2014mean,bjork2017time,dai2020log} for motivations for and examples of (\ref{Eq:Reward}). 

Conceptually, the non-Markov property implies that it is not enough to record the current state $X^\bu_t$ only. More information from $\cF_t$ is needed. For Volterra processes, the concatenated path $\omega$ is sufficient. Furthermore, by writing the reward $J$ as a functional of $\omega = X^\bu \otimes_t \Theta^{t, \bu}$ rather than $X^\bu_{t\wedge \cdot}$ only, $J$ preserves some nice regularity properties, such as continuity, under mild conditions. To clarify it further, although $\Theta^{t, \bu}$ is a functional of $X^\bu_{t\wedge \cdot}$, the dependence is usually discontinuous under uniform convergence due to the stochastic integrals involved. Readers may refer to \citet[Remark 3.2]{viens2017martingale} for a specific example. \cite{viens2017martingale} discovers that the flow property can be recovered by including $\Theta^{t, \bu}$, resulting in the functional It\^o formula. Section \ref{Sec:Example} uses specific examples to clarify the rationale more concretely. 

Let $m$ be a generic positive value for the polynomial growth rate, which may vary from line to line. By the supremum norm $||\omega||_T \triangleq \sup_{ 0 \leq t \leq T} |\omega_t|$ defined in Appendix \ref{Sec:FuncIto}, we introduce continuity in $\omega$ under $||\cdot||_T$.
\begin{assumption}
	Properties for $F$ and $G$:
	\begin{enumerate}[label={(\arabic*).}]
		\item For any fixed $s$ and $y$, $F$ is of polynomial growth in $\omega$. That is,
		\begin{align}
		|F(s, y, \omega)| \leq C_0[ 1 + ||\omega||^m_T],
		\end{align}
		for some constants $C_0, m > 0$. 
		
		\item For any fixed $s$ and $y$, $G(s, y, z)$ is continuously differentiable in $z$.
	\end{enumerate}
\end{assumption}
Similarly, for a given feedback strategy $\bu$, let
\begin{equation}
C^\bu(t, \omega, s, y)  \triangleq C(s, y, t, \omega_{t \wedge \cdot},  \bu(t, \omega_{t \wedge \cdot})).
\end{equation}

\begin{definition}\label{Def:U}
	$\bu$ is said to be an admissible strategy, denoted by $\bu \in \cU$, if 
	\begin{enumerate}[label={(\arabic*).}]
		\item Assumption \ref{Assum:SVIE} holds.
		\item (a) If $\mu^\bu$ and $\sigma^\bu$ are regular, then for a fixed $t \in [0, T]$, assume that $\mu^\bu(t; r, \omega)$ and $\sigma^\bu(t; r, \omega)$ are right-continuous in $r \in [0, t]$ and continuous in $\omega \in \Omega$.  $\partial_t \mu^\bu(t; r, \cdot)$ and $\partial_t \sigma^\bu(t; r, \cdot)$ exist for $t \in [r, T]$. For $\varphi = \mu^\bu, \sigma^\bu, \partial_t \mu^\bu, \partial_t \sigma^\bu$, 
		\begin{equation}
		|\varphi(t; r, \omega)| \leq C_0[ 1 + ||\omega||^m_T],
		\end{equation}
		for some constants $C_0, m > 0$.
		
		(b) If $\mu^\bu$ and $\sigma^\bu$ are singular, then for a fixed $t \in [0, T]$, assume that $\mu^\bu(t; r, \omega)$ and $\sigma^\bu(t; r, \omega)$ are right-continuous in $r \in [0, t)$ and continuous in $\omega \in \Omega$.  For $\varphi = \mu^\bu, \sigma^\bu$, suppose that $\partial_t \varphi (t; r, \cdot)$ exists for $t \in (r, T]$, and there exists $0 < H < 1/2$ such that for any $0 \leq r < t \leq T$,
		\begin{align}
		|\varphi(t; r, \omega)| & \leq C_0[ 1 + ||\omega||^m_T](t - r)^{H-1/2}, \\
		|\partial_t \varphi(t; r, \omega)| &\leq C_0[ 1 + ||\omega||^m_T](t - r)^{H - 3/2},
		\end{align}
		for some constants $C_0, m > 0$.
		
		\item For any fixed $s$ and $y$, $C^\bu$ is continuous in $(t, \omega)$. $C^\bu$ is of polynomial growth in $\omega$, uniformly in $t$. That is,
		\begin{equation}
		|C^\bu(t, \omega, s, y)| \leq C_0[ 1 + ||\omega||^m_T],
		\end{equation}
		for some constants $C_0, m > 0$.
		\item For any fixed $s$ and $y$, 
		\begin{align}
		\E \Big[ & \sup_{ t \leq r \leq T} \big| C(s, y, r,  X^{t, \omega, \bu}_\rdot, \bu(r, X^{t, \omega, \bu}_\rdot)) \big| + \big| F(s, y, X^{t, \omega, \bu}_{T \wedge \cdot}) \big| \Big| \cF_t \Big]  \\ 
		& + \big| G(s, y, \E[ X^{t, \omega, \bu}_T | \cF_t]) \big| \leq C_0[ 1 + ||\omega||^m_T], \nonumber
		\end{align}
		for some constants $C_0, m > 0$, which are independent of $t$.
	\end{enumerate}
\end{definition}

As a strategy that is optimal at time $t$ may no longer be optimal afterward, an agent is motivated to deviate from it. \cite{strotz1955myopia,ekeland2008investment,bjork2017time} argue that any reasonable agent should only choose strategies from which she will not deviate. Informally, the agent who is aware of time-inconsistency can consider her selves at different future times as different agents. The agent at time $t$ controls the state $X^\bu$ exactly at time $t$ by choosing a control function $\bu(t, \cdot)$. A candidate admissible equilibrium strategy $\hat \bu$ should have the following property: if for each $r > t$, the agent at time $r$ chooses $\hat{\bu}(r, \cdot)$, then it is optimal for the agent at time $t$ to choose $\hat{\bu}(t, \cdot)$. Such an equilibrium strategy is formulated backwardly by induction. Therefore, the derived $\hat{\bu}$ will not be deviated from.

The informal arguments above work well in discrete time; but in continuous time, controlling on a time set of Lebesgue measure zero has no effect on the state process. Thus, the agent at time $t$ is allowed to act on an infinitesimally small interval $[t, t+h)$ and then send $h$ to zero. Formally, let $\bu(r, \omega_\rdot)$ be a deterministic map that is also admissible. Perturbing $\hat \bu$ in the same way as \cite{bjork2014mean,bjork2017time,dai2020log,he2018equilibrium} yields the following:
\begin{align}\label{Eq:uh}
\bu_h(r, \omega_\rdot) = \left\{\begin{array}{ll}
\bu(r, \omega_\rdot), \quad t \leq r < t+h, \\
\hat \bu(r, \omega_\rdot), \quad t+h \leq r \leq T.
\end{array}\right.
\end{align}
If we denote the solution to SVIE (\ref{Eq:Volterra0}) with $\bu_h$ as $X^{\bu_h}$, the feedback strategy reads as follows:
\begin{align}
\bu_h(r, X^{\bu_h}_\rdot) = \left\{\begin{array}{ll}
\bu(r, X^{\bu_h}_\rdot), \quad t \leq r < t+h, \\
\hat \bu(r, X^{\bu_h}_\rdot), \quad t+h \leq r \leq T.
\end{array}\right.
\end{align}
A crucial characteristic of the feedback (closed-loop) formulation is that perturbing $\hat \bu$ on $[t, t+h)$ implicitly affects the strategies on $[t+h, T]$ through $X^{\bu_h}$. It is different with open-loop strategies whose value on $[t+h, T]$ is unchanged \citep{hu2012time}.

Loosely speaking, the candidate equilibrium $\hat{\bu}$ should satisfy the following property: If all agents at $[t+h, T]$ agree to use $\hat{\bu}$, then it is optimal for the agent at time $t$ to adopt $\hat{\bu}$. Mathematically, we have the following definition.

\begin{definition}\label{Def:Equilibrium}
	Consider a candidate equilibrium law $\hat \bu$. For any $t \in [0, T)$ and $\bu \in \cU$, where $\cU$ is defined in Definition \ref{Def:U}, define $\bu_h$ as in (\ref{Eq:uh}), $\hat \bu$ is an (weak) equilibrium strategy if 
	\begin{equation}\label{Eq:weak}
	\liminf_{h \downarrow 0} \frac{J(t, \omega; \hat \bu) - J(t, \omega; \bu_h)}{h} \geq 0,
	\end{equation}
	for any $\omega \in \tilde \Lambda(\hat \bu, t)$.
\end{definition}

In Definition \ref{Def:Equilibrium}, we consider a path-dependent counterpart of the concept of {\it support}. $\tilde \Lambda(\hat \bu, t)$ is the support of paths for $X^{\hat \bu} \otimes_t \Theta^{t, \hat \bu}$ conditional on $\cF_0$. The support is the set of $\omega \in \Omega$ such that any neighborhood of $\omega$ has a positive measure under the distribution of $X^{\hat \bu} \otimes_t \Theta^{t, \hat \bu}$. The metric is induced by norm $||\cdot||_T$. Roughly speaking, the support contains all possible situations for the paths. We refer to \cite{he2018equilibrium} for the rationale for considering the support rather than the whole space $\Omega$.

\begin{remark}
	Similar to \cite{he2018equilibrium}, the definition of support differs from the standard definition in the literature. We refer readers to the footnote under \citet[Definition 2]{he2018equilibrium} for further details. Characterizing the support $\tilde \Lambda(\hat \bu, t)$ under SVIE (\ref{Eq:Volterra0}) remains an open problem. Related work of which we are aware includes \cite{kalinin2019support}. However, in the examples we consider, the support is clear and relatively straightforward to obtain.
\end{remark}
\begin{remark}
	As noted in \citet[Remark 3.5]{bjork2017time}, $\hat \bu$ under (\ref{Eq:weak}) may merely be a stationary point. Recent works have also considered equilibria under the following condition:
	\begin{equation}\label{Eq:strong}
	J(t, \omega; \bu_h) \leq J(t, \omega; \hat \bu),
	\end{equation}
	where $\bu_h$ is selected in certain sets. \cite{he2018equilibrium} clarifies three notions, namely strong, regular, and weak equilibria. \cite{huang2018strong} considers a stochastic control problem in which the generator of certain Markov chains can be controlled, with a definition such as that in (\ref{Eq:strong}). However, weak equilibria should be considered first, as other types of equilibrium strategies are under stronger conditions that may be too restrictive.
\end{remark}

To emphasize that probability is also part of the weak solution, denote the expectation and conditional expectation under equilibrium control $\hat{\bu}$ by $\hatE[\cdot]$ and $\hatE[\cdot | \cF_t]$, respectively. $\hE[\cdot]$ and $\hE[\cdot | \cF_t]$ are under a perturbed control $\bu_h$ in \eqref{Eq:uh}. Therefore, in Definition \ref{Def:Equilibrium}, $J(t, \omega; \hat \bu)$ is under $\hatE[\cdot | \cF_t]$ and $J(t, \omega; \bu_h)$ is under $\hE[\cdot | \cF_t]$.

\subsection{Extended path-dependent HJB equation}\label{Sec:EP-HJB}
The following notations are useful. Define
\begin{align}
f^\bu(t, \omega, s, y) & \triangleq \E\Big[ F(s, y, X^{t, \omega, \bu}_{T \wedge \cdot}) \Big| \cF_t\Big], \\
g^\bu(t, \omega) &\triangleq  \E \big[X^{t, \omega, \bu}_T \big| \cF_t \big], \\
c^{r, \bu}(t, \omega, s, y) &\triangleq  \E\Big[ C(s, y, r,  X^{t, \omega, \bu}_\rdot, \bu(r,  X^{t, \omega, \bu}_\rdot)) \Big| \cF_t \Big]. 
\end{align}
When $\bu = \hat \bu$, denote
\begin{align}
f(t, \omega, s, y) & \triangleq  \hatE\Big[ F(s, y, X^{t, \omega, \hat \bu}_{T \wedge \cdot}) \Big| \cF_t\Big], \\
g(t, \omega) &\triangleq  \hatE \big[X^{t, \omega, \hat \bu}_T \big| \cF_t \big], \\
c^r(t, \omega, s, y) &\triangleq  \hatE\Big[ C(s, y, r,  X^{t, \omega, \hat \bu}_\rdot, \hat \bu(r,  X^{t, \omega, \hat \bu}_\rdot)) \Big| \cF_t \Big]. 
\end{align}
Our convention is that the last two arguments $(s, y)$ are reserved for state-dependence. These auxiliary functions are reduced to their counterparts in \cite{bjork2017time} when there is no path dependence. 

First, we derive a recursive relationship, which extends \citet[Lemma 3.3]{bjork2014theory} to the non-Markovian case applicable to our problem. To do so, we investigate the problem at time $t+ h$. Denote the path
\begin{align}
\omega^{t+h}_s = \left\{\begin{array}{ll}
\omega_s, &\quad 0 \leq s < t, \\
X^{t, \omega, \bu}_s, &\quad t \leq s < t+h, \\
\Theta^{t+h, \bu}_s, &\quad t+h \leq s \leq T,
\end{array}\right.
\end{align}
where
\begin{align*}
\Theta^{t+h, \bu}_s = x & + \int^{t+h}_0 \mu(s; r, X^{t, \omega, \bu}_\rdot, \bu(r, X^{t, \omega, \bu}_\rdot)) dr + \int^{t+h}_0 \sigma(s; r, X^{t, \omega, \bu}_\rdot, \bu(r, X^{t, \omega, \bu}_\rdot)) dW_r.
\end{align*}
Note that $\omega^{t+h}$ is adapted to $\cF_{t+h}$ but not $\cF_t$. To make the notation compact, we write the following:
\begin{equation}
\omega^{t+h}_s = (X^{t, \omega, \bu} \otimes_{t+h} \Theta^{t + h, \bu})_s \triangleq X^{t, \omega, \bu}_s \id_{\{0 \leq s < t+h\}} + \Theta^{t + h, \bu}_s \id_{\{t+h \leq s \leq T\}}.
\end{equation}
$\Theta^{t + h, \bu}$ is only defined on $[t+h, T]$. $X^{t, \omega, \bu}_s \neq \omega_s$ for  $ s \in (t, t+h)$ and $\Theta^{t + h, \bu}_s \neq \omega_s$ for $s \in [t+h, T]$.

\begin{lemma}\label{Lem:Recursion}
	For a general admissible feedback strategy $\bu$, the reward functional $J$ satisfies the following recursion:
	\begin{align}\label{Eq:Recursive}
	J(t, \omega; \bu) = & \E \big[ J(t+h, \omega^{t+h}; \bu) \big| \cF_t \big]  \nonumber\\ 
	& - \Big\{ \int^T_{t+h} \E \big[ c^{r, \bu}(t+h, \omega^{t+h}, t+h, \omega^{t+h}_{t+h}) \big| \cF_t \big] dr  \\ 
	& \qquad - \int^T_t \E\big[ c^{r, \bu}(t+h, \omega^{t + h}, t, \omega_t)\big| \cF_t \big] dr \Big\} \nonumber \\
	& - \big\{ \E \big[ f^\bu(t+h, \omega^{t+h}, t+h, \omega^{t+h}_{t+h}) \big| \cF_t \big] - \E\big[ f^\bu(t+h, \omega^{t+h}, t, \omega_t) \big| \cF_t \big] \big\} \nonumber \\
	&  -\big\{ \E \big[ G(t+h, \omega^{t+h}_{t+h}, g^\bu(t+h, \omega^{t+h})) \big| \cF_t \big] - G\big(t, \omega_t, \E\big[ g^\bu(t+h, \omega^{t+h}) \big| \cF_t\big] \big) \big\}. \nonumber
	\end{align}
\end{lemma}
The proof of Lemma \ref{Lem:Recursion} solely applies the tower property of conditional expectation and does not rely on the functional It\^o formula. However, the verification theorem does need the functional It\^o formula in \citet[Theorem 3.10 and Theorem 3.17]{viens2017martingale}, quoted as Theorem \ref{Thm:Ito}. The derivatives and spaces $C^{1, 2}_+( \Lambda)$ and $C^{1, 2}_{+, \alpha}(\Lambda)$ in Theorem \ref{Thm:Ito} are defined in \cite{viens2017martingale}, which are also briefly reviewed in the Appendix \ref{Sec:FuncIto}.

\begin{theorem}[{\citet[Theorem 3.10 and Theorem 3.17]{viens2017martingale}}]\label{Thm:Ito}
	Suppose that (1)-(2) of Definition \ref{Def:U} hold. Let $f \in C^{1, 2}_+( \Lambda)$ for the regular case or $f \in C^{1, 2}_{+, \alpha}( \Lambda)$ for the singular case with $\beta \triangleq \alpha + H - \frac{1}{2} > 0$. The constant $H$ is defined in Definition \ref{Def:U} (2) for the singular case. Then,
	\begin{align}
	df(t, X^\bu \otimes_t \Theta^{t, \bu}) =& \partial_t f(t, X^\bu \otimes_t \Theta^{t, \bu}) dt + \ang{ \partial_\omega f(t, X^\bu \otimes_t \Theta^{t, \bu}), \mu^{t, \bu}} dt \cr
	& + \frac{1}{2} \ang{ \partial^2_{\omega\omega} f(t, X^\bu \otimes_t \Theta^{t, \bu}), (\sigma^{t, \bu}, \sigma^{t, \bu})} dt  \cr 
	& + \ang{ \partial_\omega f(t, X^\bu \otimes_t \Theta^{t, \bu}), \sigma^{t, \bu} } d W_t, \quad \p-\as,
	\end{align}
	where for $\varphi = \mu, \sigma$, the notation $\varphi^{t, \bu}_s \triangleq \varphi^\bu(s; t, \cdot)$ emphasizes the dependence on $s \in [t, T]$.
	
	For the singular case, the derivatives related to $\omega$ are defined as follows:
	\begin{align}
	\ang{ \partial_\omega f(t, \omega), \varphi^{t, \bu}}  & \triangleq \lim_{\delta \downarrow 0} \ang{ \partial_\omega f(t, \omega), \varphi^{\delta, t, \bu}},  \quad \varphi = \mu, \sigma, \label{Eq:singular1} \\
	\ang{ \partial^2_{\omega\omega} f(t, \omega), (\sigma^{t, \bu}, \sigma^{t, \bu})} & \triangleq \lim_{\delta \downarrow 0} \ang{ \partial^2_{\omega\omega} f(t, \omega), (\sigma^{\delta, t, \bu}, \sigma^{\delta, t, \bu})}, \label{Eq:singular2}
	\end{align}
	where $\varphi^{\delta, t, \bu}_s \triangleq \varphi^\bu( s \vee (t + \delta); t, \cdot)$ for $ 0 < \delta \leq T - t$, is the truncated function. It also emphasizes the dependence on $s \in [t, T]$.
\end{theorem}

Define the value function as follows:
\begin{equation}
V(t, \omega) = J(t, \omega; \hat \bu).
\end{equation}
For a general admissible control $\bu$ and a functional $f(t, \omega)$ that satisfies Assumption \ref{Assum:ValueFunc}, as specified later, denote the operator $\bA^\bu$ as follows:
\begin{equation}\label{Eq:Operator}
(\bA^\bu f)(t, \omega) \triangleq \partial_t f(t, \omega) + \ang{ \partial_\omega f(t, \omega), \mu^{t, \bu}} + \frac{1}{2} \ang{ \partial^2_{\omega\omega} f(t, \omega), (\sigma^{t, \bu}, \sigma^{t, \bu})},
\end{equation}
where we omit the arguments in $\mu^{t, \bu}$ and $\sigma^{t, \bu}$. The derivatives in (\ref{Eq:Operator}) are defined in (\ref{Eq:t_deri}), (\ref{Eq:spatial1_deri}), and (\ref{Eq:spatial2_deri}) for regular cases, while (\ref{Eq:singular1}) and (\ref{Eq:singular2}) are for singular cases. The operator $\bA^\bu$ only applies to variables within parentheses. For instance, $ (\bA^\bu f)(t, \omega, t, \omega_t)$ operates on $t, \omega, t, \omega_t$, whereas $(\bA^\bu f^{s, y})(t, \omega)$ operates on $t, \omega$ only.

\begin{definition}\label{Def:EP-HJB}
	The extended PHJB equation system is defined as follows:
	\begin{enumerate}[label={(\arabic*).}]
		\item The function $V$ satisfies
		\begin{align}\label{Eq:V}
		\sup_{\bu \in \cU} \Big\{ & (\bA^\bu V)(t, \omega) + C(t, \omega_t, t, \omega_{t \wedge \cdot},  \bu(t, \omega_{t \wedge \cdot})) - \int^T_t (\bA^\bu c^r)(t, \omega, t, \omega_t) dr \cr
		& + \int^T_t (\bA^\bu c^{t, \omega_t, r})(t, \omega) dr - (\bA^\bu f)(t, \omega, t, \omega_t) + (\bA^\bu f^{t, \omega_t})(t, \omega) \cr
		& - \bA^\bu(G \diamond g)(t, \omega) + \partial_y G(t, \omega_t, g(t, \omega)) (\bA^\bu g)(t, \omega) \Big\} = 0, \quad 0 \leq t \leq T, \nonumber \\
		& V(T, \omega) = F(T, \omega_T, \omega) + G(T, \omega_T, \omega_T).
		\end{align} 
		Let $\hat \bu$ be the strategy that attains the supremum.
		\item For each fixed $s$ and $y$, $f^{s,y}(t, \omega)$ is defined as follows: 
		\begin{equation}\label{Eq:fsy}
		(\bA^{\hat \bu} f^{s, y})(t ,\omega) = 0, \quad 0 \leq t \leq T, \quad f^{s, y}(T, \omega) = F(s, y, \omega).  
		\end{equation}
		\item The function $g$ satisfies 
		\begin{equation}\label{Eq:g}
		(\bA^{\hat \bu} g)(t, \omega) = 0, \quad 0 \leq t \leq T, \quad g(T, \omega) = \omega_T.
		\end{equation}
		\item For each fixed $s$, $r$, and $y$, $c^{s, y, r}$ is defined by
		\begin{align}
		(\bA^{\hat \bu} c^{s, y, r}) (t, \omega) & = 0, \quad 0 \leq t \leq r, \label{Eq:csyr} \\
		c^{s, y, r}(r, \omega) & = C(s, y, r, \omega_\rdot, \hat \bu(r, \omega_\rdot)). \nonumber
		\end{align}
		\item The notations have the following meanings:
		\begin{align*}
		f(t, \omega, s, y) &= f^{s,y}(t, \omega), \qquad c^r(t, \omega, s, y) = c^{s, y, r}(t, \omega), \\
		(G \diamond g)(t, \omega) &= G(t, \omega_t, g(t, \omega)), \qquad \partial_y G(t, \omega_t, y) = \frac{\partial G}{\partial y} (t, \omega_t, y). 
		\end{align*}
		\item The probabilistic interpretations are as follows:
		\begin{align*}
		f^{s, y}(t, \omega) & = \hatE \big[ F(s, y, X^{t, \omega, \hat \bu}_{T \wedge \cdot}) \big| \cF_t \big], \qquad g(t, \omega) = \hatE \big[X^{t, \omega, \hat \bu}_T \big| \cF_t \big], \\
		c^{s, y, r}(t, \omega) &= \hatE\Big[ C(s, y, r,  X^{t, \omega, \hat \bu}_\rdot, \hat \bu(r,  X^{t, \omega, \hat \bu}_\rdot)) \Big| \cF_t \Big], \quad 0 \leq t \leq r.
		\end{align*}
	\end{enumerate}
	Equations (\ref{Eq:V})-(\ref{Eq:csyr}) above hold for $\omega \in \tilde \Lambda (\hat{\bu}, t)$, $t \in [0, T]$.
\end{definition}
\begin{remark}
	The spatial region in (\ref{Eq:V})-(\ref{Eq:csyr}) is expressed as $\tilde \Lambda (\hat{\bu}, t), \; t \in [0, T]$, which is consistent with Definition \ref{Def:Equilibrium}. For example, for MVP with state-dependent risk aversion in \cite{bjork2014mean}, the assumption is that wealth stays positive implicitly. This implies that the system in \citet[Definition 2]{bjork2014mean} holds for region $ M > 0$ instead of $ M \in \R$.
\end{remark}

We must emphasize the dependence on $\omega$ and $\omega_{t \wedge \cdot}$ in (\ref{Eq:V})-(\ref{Eq:csyr}). Although the functionals $V, c^r, c^{s, y, r}, f, f^{s, y}$, and $g$ generally depend on the whole path $\omega$, the strategy $\bu$ only depends on $\omega_{t \wedge \cdot}$, paths up to time $t$. $C(t, \omega_t, t, \omega_{t \wedge \cdot},  \bu(t, \omega_{t \wedge \cdot}))$ also only depends on $\omega_{t \wedge \cdot}$. This follows from the definition of paths in (\ref{Eq:omega}) and the fact that $\bu$ and $C$ only depend on $X^\bu$ but not $\Theta^{t, \bu}$ directly. If there is no path dependence, the system in (\ref{Eq:V})-(\ref{Eq:csyr}) reduces to the one in \cite{bjork2017time}.

We impose the regularity condition, Assumption \ref{Assum:ValueFunc}, on the functionals appearing in the extended PHJB system in Definition \ref{Def:EP-HJB}. This condition validates that all of the derivatives are well defined, although it is not the mildest condition. Indeed, we require the functionals to have spatial derivatives on $\Omega$ rather than merely on the $\tilde \Lambda (\hat{\bu}, t)$. 
\begin{assumption}\label{Assum:ValueFunc}
	For the regular case, 
	\begin{enumerate}[label={(\arabic*).}]
		\item $V, f, G\diamond g, g \in C^{1, 2}_+ (\Lambda)$;
		\item For any fixed $s$ and $y$, $f^{s, y} \in C^{1, 2}_+ (\Lambda)$;
		\item For any fixed $ r \in [0, T]$, $c^r \in C^{1, 2}_+ ([0, r] \times \Omega)$; and
		\item For any fixed $s$, $y$, and fixed $ r \in [0, T] $, $c^{s, y, r} \in C^{1, 2}_+ ([0, r] \times \Omega)$.
	\end{enumerate}
	For the singular case, let $\alpha \in (0, 1)$.
	\begin{enumerate}[label={(\arabic*).}]
		\item $V, f, G\diamond g, g \in C^{1, 2}_{+,\alpha} (\Lambda)$.
		\item For any fixed $s$ and $y$, $f^{s, y} \in C^{1, 2}_{+,\alpha} (\Lambda)$;
		\item For any fixed $ r \in [0, T] $, $c^r \in C^{1, 2}_{+,\alpha} ([0, r] \times \Omega)$; and
		\item For any fixed $s$, $y$, and any fixed $ r \in [0, T] $, $c^{s, y, r} \in C^{1, 2}_{+,\alpha} ([0, r] \times \Omega)$.
	\end{enumerate}
	In addition, $\beta \triangleq \alpha + H - \frac{1}{2} > 0$, where the constant $H$ is defined in Definition \ref{Def:U} (2) for the singular case.
\end{assumption}

We frequently encounter stochastic integrals that are required to be true martingales. Lemma \ref{Lem:TrueMartingale} is useful for the related justification. For ease of notation, we denote $\ang{\partial_\omega f, \sigma^{t, \bu}} \cdot W  \triangleq \int^\cdot_t \ang{\partial_\omega f(r, X^\bu \otimes_r \Theta^{r, \bu}), \sigma^{r, \bu}} dW_r$ for later use.
\begin{lemma}\label{Lem:TrueMartingale}
	Suppose that $\bu$ is admissible. Let $f$ be a general functional and $f \in C^{1, 2}_+ (\Lambda)$ for the regular case or $f \in C^{1, 2}_{+, \alpha} (\Lambda)$ for the singular case, with $\beta \triangleq \alpha + H - 1/2 > 0$. Then,
	\begin{equation}
	\E \Big[\int^T_t \big| \ang{\partial_\omega f(r, X^\bu \otimes_r \Theta^{r, \bu}), \sigma^{r, \bu}} \big|^2 dr \Big| \cF_t \Big] < \infty,
	\end{equation}
	which implies that $\ang{\partial_\omega f, \sigma^{t, \bu}} \cdot W$ is a true martingale.
\end{lemma}

We can now provide the verification theorem, which is one of the main results of this paper. The proof is in the same spirit of \citet[Theorem 5.2]{bjork2017time} but invokes Lemmas \ref{Lem:Recursion} and \ref{Lem:TrueMartingale} and the functional It\^o formula in \cite{viens2017martingale}. The proof consists of two steps. First, we show that the interpretations in Definition \ref{Def:EP-HJB} (6) hold and that $V(t, \omega) = J(t, \omega; \hat \bu)$. The functional It\^o formula and Lemma \ref{Lem:TrueMartingale} are applied to prove the martingale property of three functions in Definition \ref{Def:EP-HJB} (6). $V(t, \omega) = J(t, \omega; \hat \bu)$ is verified in a similar way, under the conditions in Definition \ref{Def:EP-HJB}. Second, we prove that $\hat \bu$ is indeed an equilibrium strategy under Definition \ref{Def:Equilibrium}. The recursive relationship in Lemma \ref{Lem:Recursion} with $\bu_h$ gives one representation of $J(t, \omega; \bu_h)$. The PHJB equation \eqref{Eq:V} with $\bu_h$ implies an inequality related to $V(t, \omega)$, which equals $J(t, \omega; \hat \bu)$ according to the first step. A comparison of these two sides deduces \eqref{Eq:weak} as desired. In the whole proof, including $\Theta^{t, \bu}$ recovers the flow property and overcomes the non-Markovian and non-semimartingale difficulty.

\begin{theorem}[Verification theorem]\label{Thm:Verification}
	Suppose that the extended PHJB system in \eqref{Eq:V}-\eqref{Eq:csyr} in Definition \ref{Def:EP-HJB} admits a solution $(V, f, g, f^{s, y}, c^{r}, c^{s, y, r})$ that satisfies Assumption \ref{Assum:ValueFunc}. If $\hat \bu$ realizes the supremum in \eqref{Eq:V} for $V$ and $\hat \bu$ is admissible, then $\hat \bu$ is an equilibrium law in the sense of Definition \ref{Def:Equilibrium} and $V$ is the corresponding value function.
\end{theorem}

\section{Examples}\label{Sec:Example}
In this section, we investigate the impact of volatility roughness. We apply the general framework to some specific decision-making situations with explicit or semi-closed form solutions. They include TC-MVP with constant risk aversion \citep{basak2010dynamic}, TC-MVP for log returns \citep{dai2020log}, and the MV objective with a linear controlled Volterra process. We refer to TC-MVP with constant risk aversion as the {\it const-MV} case and to TC-MVP for log returns as the {\it log-MV} case. We focus on the Volterra Heston model, which is a specific form of SVIE (\ref{Eq:Volterra0}). 

The concept of resolvent is used frequently. Kernel $R$ on $[0, \infty)$ is referred to as the {\em resolvent}, or the {\em resolvent of the second kind}, of $K$ if
\begin{equation}
K*R(t) = R*K(t) = K(t) - R(t), \quad \forall \; t \geq 0,
\end{equation}
where $*$ denotes the convolution operation:
\begin{equation}
K * R(t) = \int^t_0 K(t - s) R(s) ds, \quad \forall \; t > 0.
\end{equation}
The integral is extended to $t = 0$ by right-continuity if possible. Further properties of these definitions can be found in \cite{gripenberg1990volterra,abi2017affine}. Examples of kernels are available in Table \ref{Tab:Kernel}.

Let $R_\lambda$ be the resolvent of $\lambda K$ such that 
\begin{equation}
\lambda K * R_\lambda = R_\lambda * ( \lambda K) = \lambda K - R_\lambda.
\end{equation}
If $\lambda = 0$, $R_\lambda/\lambda = K$ and $R_\lambda = 0$.

\begin{table}[h!]
	\centering
	\begin{tabular}{c c c c }
		\hline
		& Constant	 &  Fractional (Power-law)  & Exponential\\ 
		\hline \\[0.5ex]
		$K(t)$ 	& $c$ & $c \frac{t^{\alpha-1}}{\Gamma(\alpha)}$  & $ce^{-\beta t}$ \\ \\
		$R(t)$ 		& $ce^{-ct}$  & $ct^{\alpha-1}  E_{\alpha, \alpha} (-ct^{\alpha})$ & $ce^{-\beta t}e^{-ct}$ \\ \\
		\hline
	\end{tabular}
	\caption{Examples of kernels $K$ and their resolvents $R$. $E_{\alpha,\beta}(z)=\sum_{n=0}^\infty \frac{z^n}{\Gamma(\alpha n+\beta)}$ is the Mittag-Leffler function. See \citet[Appendix A.1]{el2019characteristic} for its properties. The constant $c \neq 0$.}
	\label{Tab:Kernel}
\end{table}

Unlike variance, the wealth process (\ref{Eq:wealth}) does not have a convolution feature. Roughly speaking, certain Markov property is thus maintained. The dependence on wealth does not involve the entire trajectories. The following examples demonstrate this point.

%

\subsection{Const-MV: TC-MVP under constant risk aversion}\label{Sec:ConstMV}

Consider the TC-MVP in \cite{basak2010dynamic} under the Volterra Heston model (\ref{Eq:VoltHeston}) and wealth (\ref{Eq:wealth}):
\begin{equation}\label{Eq:ConstMV_Reward}
\E_t \left[ M^\bu_T \right] - \frac{\gamma}{2} \text{Var}_t \left[M^\bu_T\right] = \E_t \left[ M^\bu_T - \frac{\gamma}{2} (M^\bu_T)^2\right] + \frac{\gamma}{2} \left( \E_t[M^\bu_T] \right)^2,
\end{equation}
where the constant $\gamma > 0$ reflects the risk aversion level. The general reward functional in (\ref{Eq:Reward}) then becomes
\begin{equation}
F(t, \omega_t, M^{t, \omega, \bu}_{T \wedge \cdot}) = M^\bu_T - \frac{\gamma}{2} (M^\bu_T)^2, \quad G(t, \omega_t, y) = \frac{\gamma}{2} y^2.
\end{equation}

To solve the PHJB equation system in Definition \ref{Def:EP-HJB}, we highlight that $c^r = 0$ and that $f$ is not state-dependent. Consider the following Ansatz for $V$ in (\ref{Eq:V}) and $g$ in (\ref{Eq:g}). Denote the current wealth at time $t$ by $M$ and recall $\Theta^t_s$ as defined in (\ref{Eq:VolTheta}):
\begin{eqnarray}
V(t, \omega) &=& V\big(t, M, \Theta^t_{[t, T]}\big)  = V_1(t) M + \int^T_t V_2(s) \Theta^t_s ds + V_0(t), \label{Eq:ConstMV_Ansatz_V} \\
g(t, \omega) &=& g \big(t, M, \Theta^t_{[t, T]} \big) = g_1(t) M + \int^T_t g_2(s) \Theta^t_s ds + g_0(t), \label{Eq:ConstMV_Ansatz_g}
\end{eqnarray}
where $V_1$, $V_0$, $g_1$, and $g_0$ are deterministic continuously differentiable functions and $V_2$ and $g_2$ satisfy suitable integrability conditions. This Ansatz implies that the functions $V$ and $g$ depend on current wealth $M$ and $\Theta^t_{[t, T]}$ only. 

As $V$ and $g$ are linear functionals of $\Theta^t_{[t, T]}$, the direct calculation proceeds as follows:
\begin{align}
(\bA^\bu V)(t, \omega) =& \dot{V}_1(t) M - V_2(t) \nu + \dot{V}_0(t) \\
& + (\varUpsilon M + \theta \sqrt{\nu} u) V_1(t) + (\kappa\phi - \kappa \nu) \int^T_t V_2(s)K(s-t)ds, \nonumber \\
(\bA^\bu g)(t, \omega) =& \dot{g}_1(t) M - g_2(t) \nu + \dot{g}_0(t) \\
&  + (\varUpsilon M + \theta \sqrt{\nu} u) g_1(t) + (\kappa\phi - \kappa \nu) \int^T_t g_2(s)K(s-t)ds, \nonumber \\
\bA^\bu(G \diamond g)(t, \omega) =& \gamma g \big(t, M, \Theta^t_{[t, T]} \big) (\bA^\bu g)(t, \omega) + \frac{\gamma}{2} g^2_1(t) u^2 \nonumber  \\
& + \sigma \rho u \sqrt{\nu} \gamma g_1(t) \int^T_t g_2(s) K(s-t) ds + \frac{\gamma}{2} \Big( \int^T_t g_2(s) K(s-t) ds \Big)^2 \sigma^2 \nu, \\
\partial_y G(t, \omega_t, g(t, &\omega)) = \gamma g \big(t, M, \Theta^t_{[t, T]} \big).
\end{align}
We use the fact that $\omega^\nu$ is continuous at time $t$ and $\Theta^t_t = \nu_t \triangleq \nu$.

Equation (\ref{Eq:V}) in Definition \ref{Def:EP-HJB} becomes
\begin{align}\label{Eq:HJB-V}
\sup_{\bu \in \cU} \Big\{ & \dot{V}_1(t) M - V_2(t) \nu + \dot{V}_0(t) + (\varUpsilon M + \theta \sqrt{\nu} u) V_1(t) + (\kappa\phi - \kappa \nu) \int^T_t V_2(s)K(s-t)ds \cr
& - \frac{\gamma}{2} g^2_1(t) u^2 -  \sigma \rho u \sqrt{\nu} \gamma g_1(t) \int^T_t g_2(s) K(s-t) ds \cr
& - \frac{\gamma}{2} \Big( \int^T_t g_2(s) K(s-t) ds \Big)^2 \sigma^2 \nu \Big\} = 0, \quad V_1(T) = 1, \; V_0(T) = 0.
\end{align}
Therefore,
\begin{equation}
\hat \bu(t, \nu_t) = \frac{\theta V_1(t) - \gamma \sigma \rho g_1(t) \int^T_t g_2(s)K(s-t)ds}{\gamma g^2_1(t)} \sqrt{\nu_t}. 
\end{equation}
Furthermore, we have $g_1(T) = 1$, $g_0(T) = 0$, and
\begin{align}\label{Eq:Ag}
(\bA^{\hat \bu} g)(t, \omega) =& \dot{g}_1(t) M - g_2(t) \nu + \dot{g}_0(t) + \varUpsilon M g_1(t) + \frac{\theta^2 \nu V_1(t)}{\gamma g_1(t)}  - \rho \sigma \theta \nu \int^T_t g_2(s)K(s-t)ds \nonumber \\
& + (\kappa\phi - \kappa \nu ) \int^T_t g_2(s)K(s-t)ds = 0.  
\end{align}

By separation of variables and recognizing that $g_1(t) = V_1(t)$ from (\ref{Eq:HJB-V}) and (\ref{Eq:Ag}), we obtain the following:
\begin{align}
& \dot{g}_1(t) + \varUpsilon_t g_1(t) = 0, \quad g_1(T) = 1, \label{Eq:g1}\\
& g_2(t) + (\kappa + \rho \sigma \theta) \int^T_t g_2(s)K(s-t)ds - \frac{\theta^2}{\gamma} = 0, \label{Eq:g2} \\ 
& \dot{g}_0(t) + \kappa \phi \int^T_t g_2(s)K(s-t)ds = 0, \quad g_0(T) = 0 \label{Eq:g0}
\end{align}
and
\begin{align}
& \dot{V}_1(t) + \varUpsilon_t V_1(t) = 0, \quad V_1(T) = 1, \label{Eq:ConstMV_V1}\\
& V_2(t) + \kappa \int^T_t V_2(s)K(s-t)ds + \frac{\gamma \sigma^2}{2} \Big( \int^T_t g_2(s) K(s-t) ds \Big)^2 \cr
&\qquad - \frac{\left(\theta - \gamma \sigma \rho \int^T_t g_2(s)K(s-t)ds\right)^2}{2 \gamma}= 0, \label{Eq:ConstMV_V2} \\ 
& \dot{V}_0(t) + \kappa \phi \int^T_t V_2(s)K(s-t)ds = 0, \quad V_0(T) = 0. \label{Eq:ConstMV_V0}
\end{align}

The system in (\ref{Eq:g1})--(\ref{Eq:ConstMV_V0}) can be solved explicitly. First, $g_1(t) = V_1(t) = e^{ \int^T_t \varUpsilon_sds}$. (\ref{Eq:g2}) is a linear Volterra integral equation (VIE). Existence and uniqueness can be determined using \citet[Equation (1.3), p.77]{gripenberg1990volterra}. Let $\lambda = \kappa + \rho \sigma \theta$ and recall that $R_\lambda$ is the resolvent of $\lambda K$:
\begin{equation}
g_2(t) = \frac{\theta^2}{\gamma} - \frac{\theta^2}{\gamma} \int^T_t R_\lambda(s-t)ds.
\end{equation}
Additionally, a useful result is
\begin{equation}
\int^T_t g_2(s)K(s-t)ds = \frac{\theta^2}{\gamma} \int^T_t \frac{R_\lambda(s-t)}{\lambda} ds.
\end{equation}
$g_0$ can then be solved directly. $V_2$ in (\ref{Eq:ConstMV_V2}) is also a linear VIE that can be solved in the same way as $g_2$. $V_0$ is solved after $V_2$. Although the calculation is straightforward, it is lengthy. As such, we omit it here. Ultimately,
\begin{equation}\label{Eq:ConstMV_u}
\hat \bu(t, \nu_t) = \frac{\theta}{\gamma} e^{-\int^T_t \varUpsilon_sds} \sqrt{\nu_t} -  \frac{\rho \sigma \theta^2}{\gamma} e^{-\int^T_t \varUpsilon_sds} \int^T_t \frac{R_\lambda(s-t)}{\lambda} ds \sqrt{\nu_t}. 
\end{equation}

Then the support for the wealth process is $\R$. The first term in  (\ref{Eq:ConstMV_u}) is the same as the constant volatility case. The second term can be interpreted as a hedge for the randomness from stochastic volatility. Roughness alters the hedge through the resolvent $R_\lambda$.

Verifying that $\hat \bu$ in (\ref{Eq:ConstMV_u}) is admissible in the sense of Definition \ref{Def:U} is straightforward. Indeed, Assumption \ref{Assum:SVIE} holds in view of the moment estimates in \citet[Lemma 3.1]{abi2017affine} for $\nu$. Other requirements in Definition \ref{Def:U} are direct. We summarize the analysis above in the following lemma.
\begin{lemma} \label{Lem:ConstMV}
	Problem (\ref{Eq:ConstMV_Reward}) under the Volterra Heston model (\ref{Eq:VoltHeston}) has an equilibrium strategy $\hat \bu$ given by (\ref{Eq:ConstMV_u}), which is admissible in the sense of Definition \ref{Def:U}. The value function is given by (\ref{Eq:ConstMV_Ansatz_V}), with $V_1, V_2$, and $V_0$ given by (\ref{Eq:ConstMV_V1}), (\ref{Eq:ConstMV_V2}), and (\ref{Eq:ConstMV_V0}), respectively. $g$ in (\ref{Eq:g}) is given by (\ref{Eq:ConstMV_Ansatz_g}), with $g_1, g_2$, and $g_0$ given by (\ref{Eq:g1}), (\ref{Eq:g2}), and (\ref{Eq:g0}), respectively.
\end{lemma}

\subsection{Log-MV: TC-MVP for log returns}\label{Sec:LogMV}
Instead of considering preferences for terminal wealth, \cite{dai2020log} argues that an analysis based on log returns is more plausible. The derived equilibrium strategy is wealth-dependent and will not short sell risky assets with positive excess returns over the long-term horizon. Suppose that the proportional amount of wealth in the stock is $\bpi$. $M^{\bpi}$ is the wealth process corresponding to $\bpi$. Denote $L_t^{\bpi} \triangleq \ln M^{\bpi}_t$. For ease of notation, we write $L_t = L_t^{\bpi}$. It follows that 
\begin{equation}\label{Eq:LogRet}
d L_t = [\varUpsilon_t + \theta \nu_t \pi_t - \frac{1}{2} \pi^2_t \nu_t ] dt + \sqrt{\nu_t} \pi_t dW_{1t}.
\end{equation}

Consider the TC-MVP in \cite{dai2020log} under the Volterra Heston model (\ref{Eq:VoltHeston}) and log return (\ref{Eq:LogRet}):
\begin{equation}\label{Eq:LogMV_Reward}
\E_t \left[ L_T \right] - \frac{\gamma}{2} \text{Var}_t \left[L_T\right] = \E_t \left[ L_T - \frac{\gamma}{2} (L_T)^2\right] + \frac{\gamma}{2} \left( \E_t[L_T] \right)^2.
\end{equation}

With a slight abuse of notation, we try the following Ansatz for $V$ in (\ref{Eq:V}) and for $g$ in (\ref{Eq:g}):
\begin{eqnarray}
V(t, \omega) &=& V\big(t, L, \Theta^t_{[t, T]}\big)  = L + \int^T_t V_2(s) \Theta^t_s ds + V_0(t), \label{Eq:LogMV_Ansatz_V} \\
g(t, \omega) &=& g \big(t, L, \Theta^t_{[t, T]} \big) = L + \int^T_t g_2(s) \Theta^t_s ds + g_0(t), \label{Eq:LogMV_Ansatz_g}
\end{eqnarray}
where $V_0$ and $g_0$ are deterministic continuously differentiable functions and $V_2$ and $g_2$ satisfy the suitable integrability conditions.

Equation (\ref{Eq:V}) in Definition \ref{Def:EP-HJB} becomes
\begin{align}\label{Eq:LogHJB-V}
\sup_{\bpi \in \cU} \Big\{ & - V_2(t) \nu + \dot{V}_0(t) + \varUpsilon + \theta \nu \pi - \frac{1}{2} \nu \pi^2 + (\kappa\phi - \kappa \nu) \int^T_t V_2(s)K(s-t)ds \cr
& - \frac{\gamma}{2} \nu \pi^2 -   \gamma \rho \sigma \nu \pi \int^T_t g_2(s) K(s-t) ds \cr
& - \frac{\gamma}{2} \Big( \int^T_t g_2(s) K(s-t) ds \Big)^2 \sigma^2 \nu \Big\} = 0, \quad V_0(T) = 0.
\end{align}
Therefore,
\begin{equation}
\hat \pi_t = \frac{\theta - \gamma \rho \sigma \int^T_t g_2(s)K(s-t)ds}{1 + \gamma}. 
\end{equation}
Furthermore, $(\bA^{\hat \bpi} g)(t, \omega)=0$ gives rise to
\begin{align}\label{Eq:LogAg}
& - g_2(t) \nu + \dot{g}_0(t) + \varUpsilon +  \frac{\theta \nu}{1 + \gamma} \big[\theta - \gamma \rho \sigma \int^T_t g_2(s)K(s-t)ds \big]   \nonumber \\
&  -\frac{\nu}{2(1+\gamma)^2} \big[\theta - \gamma \rho \sigma \int^T_t g_2(s)K(s-t)ds\big]^2 + (\kappa\phi - \kappa \nu ) \int^T_t g_2(s)K(s-t)ds = 0, \nonumber \\
&\quad g_0(T) = 0. 
\end{align}

By separation of variables, we obtain
\begin{align}
g_2(t) = & \frac{(1+2\gamma)\theta^2}{2(1+\gamma)^2} - \big[\kappa + \frac{\gamma^2 \rho \sigma \theta}{(1+\gamma)^2} \big] \int^T_t g_2(s)K(s-t)ds \nonumber \\
& - \frac{\gamma^2\rho^2\sigma^2}{2(1+\gamma)^2}  \big(\int^T_t g_2(s)K(s-t)ds \big)^2. \label{Eq:LogMVg2}
\end{align}
Let $\psi(T-t) = \int^T_t g_2(s)K(s-t)ds$, then convolve both sides of (\ref{Eq:LogMVg2}) with kernel $K(\cdot)$ and change $T-t$ to $t$. This yields the following Riccati-Volterra equation:
\begin{equation}\label{Eq:LogMVpsi}
\psi(t) = \int^t_0 K(t - s) \Big[ - \frac{\gamma^2\rho^2\sigma^2}{2(1+\gamma)^2}  \psi^2(s) - \big(\kappa + \frac{\gamma^2 \rho \sigma \theta}{(1+\gamma)^2} \big) \psi(s) + \frac{(1+2\gamma)\theta^2}{2(1+\gamma)^2} \Big] ds.
\end{equation}

Furthermore,
\begin{align}
& V_2(t) + \kappa \int^T_t V_2(s)K(s-t)ds + \frac{\gamma \sigma^2}{2}  \psi^2(T-t) - \frac{1}{2(1+\gamma)} \left[\theta - \gamma \rho \sigma \psi(T-t)\right]^2 = 0, \label{Eq:LogMV_V2} \\ 
& \dot{V}_0(t) + \varUpsilon_t + \kappa \phi \int^T_t V_2(s)K(s-t)ds = 0, \quad V_0(T) = 0, \label{Eq:LogMV_V0} \\
& \dot{g}_0(t) + \varUpsilon_t + \kappa \phi \psi(T-t) = 0, \quad g_0(T) = 0. \label{Eq:LogMVg0}
\end{align}

\begin{corollary}\label{Cor:Logpsi}
	Suppose that $\kappa + \frac{\gamma^2 \rho \sigma \theta}{(1+\gamma)^2} > 0$. Then, (\ref{Eq:LogMVpsi}) has a unique global continuous solution on $[0, T]$. Define 
	\begin{equation}
	H(w) \triangleq \frac{\gamma^2\rho^2\sigma^2}{2(1+\gamma)^2}  w^2 - \big(\kappa + \frac{\gamma^2 \rho \sigma \theta}{(1+\gamma)^2} \big) w - \frac{(1+2\gamma)\theta^2}{2(1+\gamma)^2} \triangleq H_2 w^2 +  H_1 w + H_0.
	\end{equation}
	Then,
	\begin{equation}
	0 < \psi(t) \leq -r_1(t) < - w_* , \quad \forall \; t > 0,
	\end{equation}
	with $ w_* = \frac{-H_1 - \sqrt{H_1^2 - 4H_2H_0}}{2H_2} < 0$ and $r_1(t) \triangleq Q^{-1}_1 \big( \int^t_0 K(s) ds \big)$, where $Q_1(w) = -\int^0_w \frac{du}{H(u)}$.
	
		In addition, system (\ref{Eq:LogMV_V2})-(\ref{Eq:LogMVg0}) has a unique continuous solution $(V_2, V_0, g_0)$ on $[0, T]$.
	
\end{corollary}

Ultimately, an equilibrium strategy is expressed as follows:
\begin{equation}\label{Eq:LogMV_pi}
\hat \pi_t = \frac{\theta}{1 + \gamma} - \frac{\gamma \rho \sigma }{1 + \gamma} \psi(T-t).
\end{equation}

For the admissibility of $\hat{\bpi}$, we have the following result about Assumption \ref{Assum:SVIE}.
\begin{corollary}\label{Cor:LogMoment}
	Assume that (\ref{Eq:LogMVpsi}) has a unique continuous solution on $[0, T]$. Suppose that
	\begin{equation}\label{Assum:LogExpMon}
	\hatE \Big[ e^{c \int^T_0 \nu_s ds} \Big] < \infty,
	\end{equation}
	with constant $c$ expressed as follows:
	\begin{equation}\label{Eq:Const_a}
	c = \max \Big\{ 2 p |\theta| \sup_{t \in [0, T]} | \hat{\pi}_t |, (8p^2 - 2p) \sup_{t \in [0, T]}  \hat{\pi}^2_t \Big\}, \quad \text{for certain } p > 1. 
	\end{equation}
	Then, wealth $M^*$ under $\hat{\bpi}$ satisfies 
	\begin{equation}\label{Eq:X*Integral}
	\hatE \Big[ \sup_{ t \in [0, T]} | M^*_t |^p \Big] < \infty.
	\end{equation}
\end{corollary}
Thus, we have the following lemma.
\begin{lemma} \label{Lem:LogMV}
	Problem (\ref{Eq:LogMV_Reward}) under the Volterra Heston model (\ref{Eq:VoltHeston}) has an equilibrium strategy $\hat \bpi$ given by (\ref{Eq:LogMV_pi}). If Assumption (\ref{Assum:LogExpMon}) holds for a large enough constant $c>0$, then the equilibrium strategy (\ref{Eq:LogMV_pi}) is admissible in the sense of Definition \ref{Def:U}. The value function is given by (\ref{Eq:LogMV_Ansatz_V}), with $V_2$ and $V_0$ given by (\ref{Eq:LogMV_V2}) and (\ref{Eq:LogMV_V0}), respectively. $g$ in (\ref{Eq:g}) is given by (\ref{Eq:LogMV_Ansatz_g}), with $g_2$ and $g_0$ given by (\ref{Eq:LogMVg2}) and (\ref{Eq:LogMVg0}), respectively.
\end{lemma}
\begin{remark}\label{Rem:FracConst}
	For the specific fractional kernel $K(t) = \frac{t^{H-1/2}}{\Gamma(H+1/2)}$, Assumption (\ref{Assum:LogExpMon}) holds for a large enough constant $c>0$ when the time horizon $T$ is sufficiently small. See \cite{gerhold2019moment}.
\end{remark}
\begin{remark}\label{Rem:GenHeston}
	For the general Heston specification considered in \cite{basak2010dynamic,dai2020log}, 
	\begin{equation}\label{Eq:Genstock}
	dS_t = S_t (\varUpsilon_t + \theta \nu^{\frac{1+\delta}{2\delta}}_t) dt + S_t \nu^{\frac{1}{2\delta}}_t dW_{1t}.
	\end{equation}
	An equilibrium strategy is
	\begin{equation}
	\hat \pi_t = \Big[ \frac{\theta}{1 + \gamma} - \frac{\gamma \rho \sigma }{1 + \gamma} \psi(T-t) \Big] \nu^{\frac{\delta - 1}{2\delta}}_t.
	\end{equation}
	The related proof and verification of admissibility are identical.
\end{remark}

\subsection{TC-MV under linear controlled Volterra processes}\label{Sec:LC-MV}
	In this section, we consider the same objective \eqref{Eq:LC-MVreward} as in Section \ref{Sec:ConstMV}: 
	\begin{equation}\label{Eq:LC-MVreward}
	\E_t \left[ X^\bu_T - \frac{\gamma}{2} (X^\bu_T)^2\right] + \frac{\gamma}{2} \left( \E_t[X^\bu_T] \right)^2.
	\end{equation}
	However, the state process $X^\bu$ is a one-dimensional linear controlled Volterra process given by \eqref{Eq:LC-MVstate}. $B, b, D$, and $\sigma$ are one-dimensional essentially bounded and deterministic measurable functions:
	\begin{equation}\label{Eq:LC-MVstate}
	X^\bu_t = x + \int^t_0 K(t-r)(B_r u_r + b_r)dr + \int^t_0 K(t-r)(D_r u_r + \sigma_r) dW_r.
	\end{equation}
	
	Consider the following Ansatz for $V$ and $g$. $\Theta^{t, \hat{\bu}}$ is defined as in \eqref{Eq:Theta} but with $\bu = \hat{\bu}$.
	\begin{eqnarray}
	V(t, \omega) &=& V_1(t) \Theta^{t, \hat{\bu}}_T + V_0(t), \label{Eq:LC-MV_Ansatz_V} \\
	g(t, \omega) &=& g_1(t) \Theta^{t, \hat{\bu}}_T + g_0(t), \label{Eq:LC-MV_Ansatz_g}
	\end{eqnarray}
	where $V_1, V_0, g_1$, and $g_0$ are deterministic continuously differentiable functions.
	
	Similarly, we have
	\begin{align*}
	(\bA^\bu V)(t, \omega) =& \dot{V}_1(t) \Theta^{t, \hat{\bu}}_T + \dot{V}_0(t) + V_1(t) K(T-t) (B_t u_t + b_t), \\
	(\bA^\bu g)(t, \omega) =& \dot{g}_1(t) \Theta^{t, \hat{\bu}}_T + \dot{g}_0(t) + g_1(t) K(T-t) (B_t u_t + b_t), \\
	\bA^\bu(G \diamond g)(t, \omega) =& \gamma g(t, \omega) (\bA^\bu g)(t, \omega) + \frac{\gamma}{2} g^2_1(t) K^2(T - t) (D_t u_t + \sigma_t)^2.
	\end{align*}
	
	Equation (\ref{Eq:V}) in Definition \ref{Def:EP-HJB} is
	\begin{align}\label{Eq:LC-HJB}
	\sup_{\bu \in \cU} \Big\{ & \dot{V}_1(t) \Theta^{t, \hat{\bu}}_T + \dot{V}_0(t) + V_1(t) K(T-t) (B_t u_t + b_t) \cr
	& - \frac{\gamma}{2} g^2_1(t) K^2(T - t) (D_t u_t + \sigma_t)^2  \Big\} = 0, \quad V_1(T) = 1, \; V_0(T) = 0.
	\end{align}
	Along with $(\bA^{\hat \bu} g)(t, \omega) =0$, we obtain the following:
	\begin{align}
	V_1(t) & = g_1(t) = 1, \\
	\dot{V}_0(t) & + \frac{(B_t - \gamma K(T-t) D_t \sigma_t)^2}{2 \gamma D^2_t} + K(T-t)b_t - \frac{\gamma}{2} K^2(T-t)\sigma^2_t = 0, \quad V_0(T) = 0, \label{Eq:LC-V0} \\
	\dot{g}_0(t) & + \frac{B^2_t - \gamma K(T-t) D_t B_t \sigma_t}{\gamma D^2_t} + K(T-t) b_t = 0, \quad g_0(T) = 0. \label{Eq:LC-g0}
	\end{align}
	An equilibrium control is
	\begin{equation}\label{Eq:LC-u}
	\hat \bu_t = \frac{B_t - \gamma K(T- t) D_t \sigma_t}{ \gamma K(T- t) D^2_t}. 
	\end{equation}
	When $K= \text{id}$, the solution reduces to that found in \citet[Section 4]{hu2012time}. Verifying the admissibility of \eqref{Eq:LC-u} is straightforward, as follows:
	\begin{lemma}
		Suppose that Assumption \ref{Assum:kernel} holds and $D^2_t \geq \delta >0$. Then, \eqref{Eq:LC-V0} and \eqref{Eq:LC-g0} admit unique solutions on $[0, T]$. \eqref{Eq:LC-u} is an admissible equilibrium strategy in the sense of Definition \ref{Def:U} for the MV objective \eqref{Eq:LC-MVreward} under the state \eqref{Eq:LC-MVstate}. 
	\end{lemma}

\section{Numerical analysis}\label{Sec:Num}
In this section, we numerically analyze the effects of roughness on stock demand and improvements on profitability. We concentrate on the rough Heston model \citep{el2019characteristic}, namely kernel $K(t) = \frac{t^{H-1/2}}{\Gamma(H+1/2)}$, $H \in (0, 1/2]$. We conduct this numerical study from three perspectives. First, we perform a sensitivity analysis by varying the roughness only and fixing the other parameters. However, roughness can have sophisticated interactions with other variables. Thus, second, we further consider a calibration situation. Third, we conduct an empirical study based on recent market data. Strategies with rough volatility outperform their classic counterparts in terms of profit and the Sharpe ratio. We mainly compare four strategies: the const-MV case (\ref{Eq:ConstMV_u}), the log-MV case (\ref{Eq:LogMV_pi}), the pre-committed MVP \citep{han2019mean}, and Merton's portfolio problem with power utility (CRRA) \citep{han2019merton}.

The interaction between time-inconsistency and volatility roughness can be complicated. We answer three questions. First, when volatility is rougher, should investors increase or decrease their stock demand? Second, do investors with different levels of risk aversion behave differently with respect to volatility roughness? Third, is it more profitable to incorporate rough volatility?

\subsection{Sensitivity analysis}
\subsubsection{Const-MV}
The equilibrium strategy in Lemma \ref{Lem:ConstMV} can be further simplified for the fractional kernel $K(t) = \frac{t^{\alpha-1}}{\Gamma(\alpha)}$ with $\alpha = H + 1/2$. For $\lambda  = 0$,
\begin{align}\label{Eq:R_lambda_0}
\int^T_t \frac{R_\lambda(s-t)}{\lambda} ds = \int^T_t K(s - t) ds = \frac{(T - t)^\alpha}{\Gamma(\alpha + 1)}. 
\end{align}
For $\lambda \neq 0$, the property of Mittag-Leffler functions in \citet[Appendix A.1]{el2019characteristic} shows that
\begin{align}\label{Eq:R_lambda}
\int^T_t \frac{R_\lambda(s-t)}{\lambda} ds = \frac{1 - E_{\alpha, 1} ( - \lambda (T - t)^\alpha)}{\lambda} \triangleq \frac{F^{\alpha, \lambda}(T - t)}{\lambda}. 
\end{align}
This enables us to compute the equilibrium strategy using an explicit formula. 

Figure (\ref{Fig:ConstMVHedge}) illustrates the plot for the hedging term,
\begin{equation}
-  \frac{\rho \sigma \theta^2}{\gamma} e^{-\int^T_t \varUpsilon_sds} \int^T_t \frac{R_\lambda(s-t)}{\lambda} ds.
\end{equation}
Note that we set a negative correlation between stock price and volatility in the plot to reflect the leverage effect. Figure \ref{Fig:ConstMV} numerically shows the following phenomenon. When the investment horizon is long (i.e., $t$ is small), the const-MV strategy advocates investing more if the stock is smoother. However, near the end of the investment horizon, the const-MV strategy prefers to invest more if the stock is rougher. We refer to this phenomenon as the {\it investment horizon effect}. It is distinct from the pre-committed MVP strategy obtained in \cite{han2019mean}. When stock volatility is smooth, the equilibrium strategy reduces the stock position gradually until the end of the investment period. In contrast, when stock volatility is rough, the equilibrium strategy suggests a relatively steady holding of the stock for a sufficiently long investment horizon but rapidly slashes the holding near the end of the investment horizon. The latter phenomenon also occurs in optimal investment problems with position limits, but our problem has no constraints on the position. In fact, the investment horizon effect can be shown mathematically by deriving asymptotic estimates of (\ref{Eq:R_lambda_0}) and (\ref{Eq:R_lambda}). We refer to the following corollary.

\begin{corollary}\label{Cor:F}
	Suppose that $\alpha \in (\frac{1}{2}, 1)$. Then, for sufficiently large values of $T - t$, integral $\int^T_t \frac{R_\lambda(s-t)}{\lambda} ds$ is increasing on $\alpha$. For sufficiently small values of $T - t$, $\int^T_t \frac{R_\lambda(s-t)}{\lambda} ds$ is decreasing on $\alpha$.
\end{corollary}

\begin{figure}
	\centering
	\begin{minipage}{0.5\textwidth}
		\centering
		\includegraphics[width=0.9\textwidth]{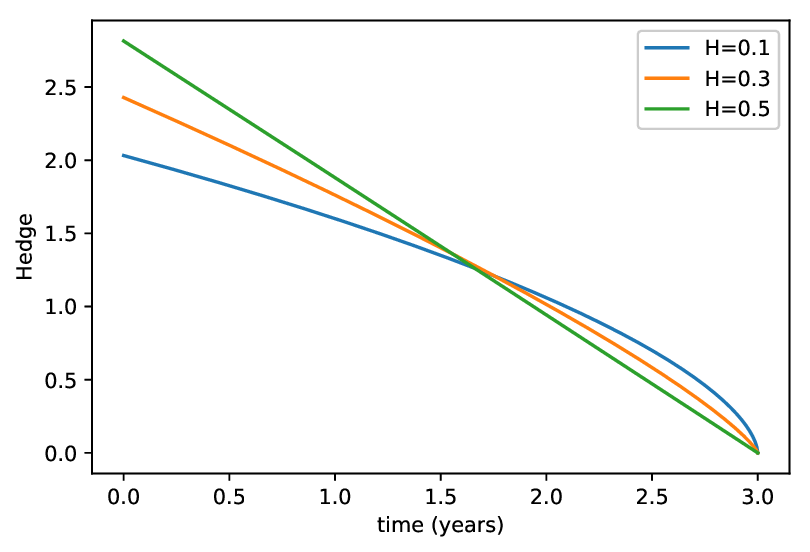}
		\subcaption{Hedge}\label{Fig:ConstMVHedge}
	\end{minipage}%
	\begin{minipage}{0.5\textwidth}
		\centering
		\includegraphics[width=0.9\textwidth]{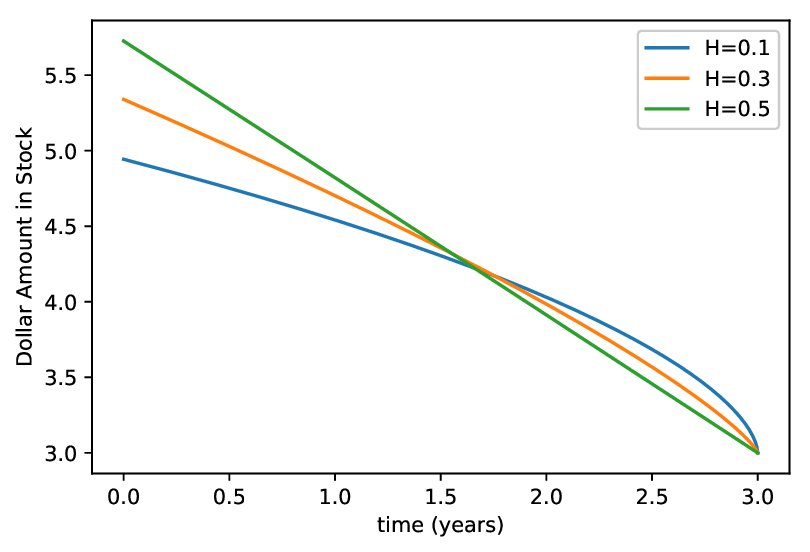}
		\subcaption{Dollar amount in stock}\label{Fig:ConstMVDollar}
	\end{minipage}
	\caption{Hedge term and dollar amount of the const-MV strategy. We set volatility-of-volatility $\sigma = 0.3$, mean reversion speed $\kappa = 0.3$, risk premium parameter $\theta = 1.5$, correlation $\rho = -0.7$, investment horizon $T = 3$, risk-free rate $\varUpsilon = 0.01$, and risk aversion $\gamma = 0.5$. $H = 0.5$ corresponds to the classic Heston model case.}\label{Fig:ConstMV}
\end{figure}

The behavior in Figure \ref{Fig:ConstMV} can be further enhanced using different investment horizons. In Figure \ref{Fig:ConstMVTime}, an investor with a relatively short horizon (e.g., 1 year) would simply have more stock demands if volatility is rougher. If the investment horizon is long enough (e.g., 10 years) and volatility is rougher, the investor purchases less for the first 8 years and purchases more for the remaining time. The diminished amount of wealth invested in the stock can be up to 40\% of the classic Heston counterpart. If other variables do not change, increments in volatility roughness reduce the willingness to retain more stock exposure in the long term. Roughness does have a material impact on the investment decisions of the investors with long planning horizons.
\begin{figure}[!h]
	\centering
	\begin{minipage}{0.5\textwidth}
		\centering
		\includegraphics[width=0.9\textwidth]{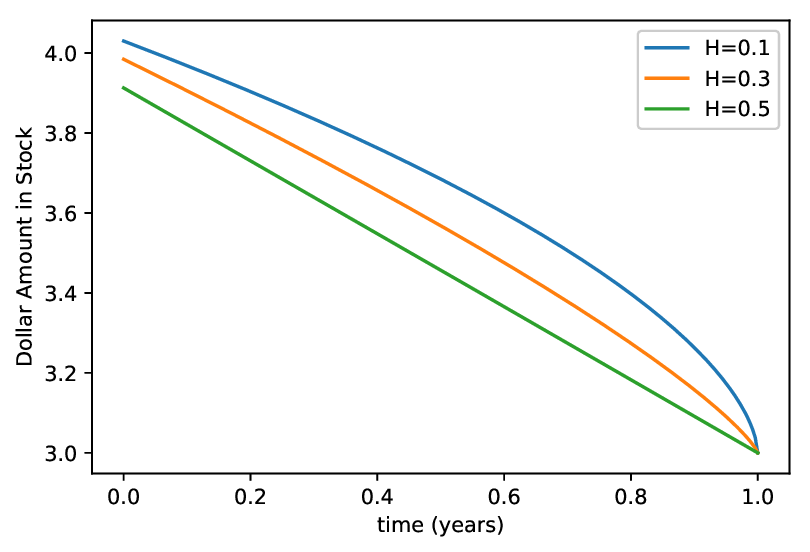}
		\subcaption{Short investment horizon}\label{Fig:ConstMVShort}
	\end{minipage}%
	\begin{minipage}{0.5\textwidth}
		\centering
		\includegraphics[width=0.9\textwidth]{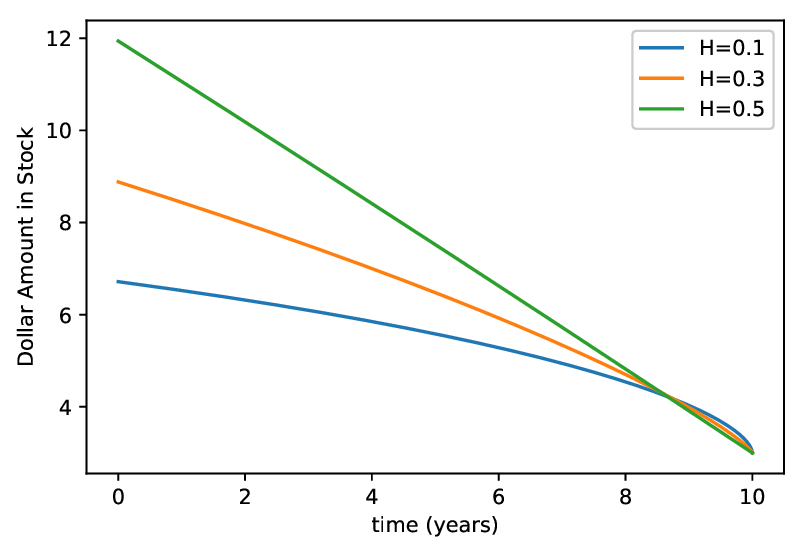}
		\subcaption{Long investment horizon}\label{Fig:ConstMVLong}
	\end{minipage}
	\caption{Dollar amount of the const-MV strategy with different investment horizons. We set all of the other parameters as in Figure \ref{Fig:ConstMV}.}\label{Fig:ConstMVTime}
\end{figure}

\subsubsection{Log-MV}
For the fractional kernel, the Riccati-Volterra equation (\ref{Eq:LogMVpsi}) can be solved numerically using the fractional Adams method, which is detailed in \citet[Section 5.1]{el2019characteristic}. The assumption in Corollary \ref{Cor:Logpsi} is validated for settings in Figures \ref{Fig:LogMV} and \ref{Fig:PreRough}. Assumption (\ref{Assum:LogExpMon}) is satisfied by Remark \ref{Rem:FracConst}.

The investment horizon effect is also observed in log-MV equilibrium strategies. When volatility is rougher, investors purchase less at first and more later on, as shown in Figure \ref{Fig:LogMV}. This is consistent with the const-MV counterpart. However, there is an important difference between the two criteria. For the const-MV case, the moment to prefer rough, shown as the interactions between the curves in Figure \ref{Fig:ConstMV}, is not affected by heterogeneity in risk aversion $\gamma$. This is not the case for the log-MV case. As shown in Figure \ref{Fig:PreRough}, when log-MV investors are more risk averse (i.e., the risk aversion parameter $\gamma$ is larger), they prefer rough earlier. This indicates that if only roughness varies and increases, then preferring rough earlier can minimize risk.
\begin{figure}
	\centering
	\begin{minipage}{0.5\textwidth}
		\centering
		\includegraphics[width=0.9\textwidth]{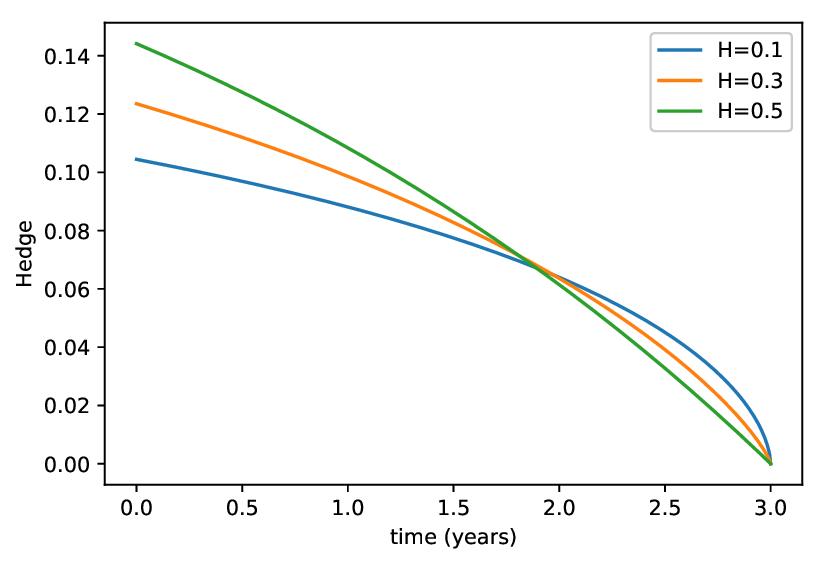}
		\subcaption{Hedge}\label{Fig:LogMVHedge}
	\end{minipage}%
	\begin{minipage}{0.5\textwidth}
		\centering
		\includegraphics[width=0.9\textwidth]{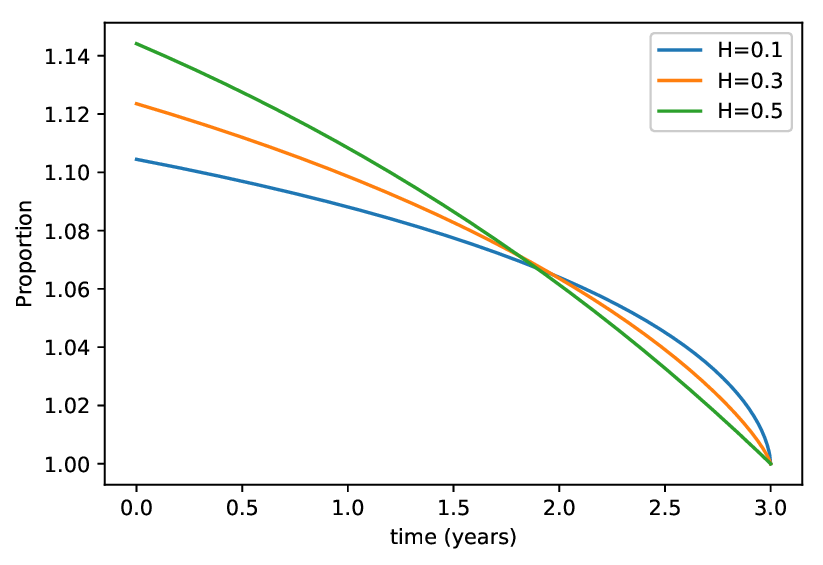}
		\subcaption{Proportion}\label{Fig:LogMVPo}
	\end{minipage}
	\caption{Hedge term and proportional amount of wealth in stock for the log-MV strategy. All of the parameters are the same as in Figure \ref{Fig:ConstMV}.}\label{Fig:LogMV}
\end{figure}

\begin{figure}
	\centering
	\includegraphics[width=0.5\textwidth]{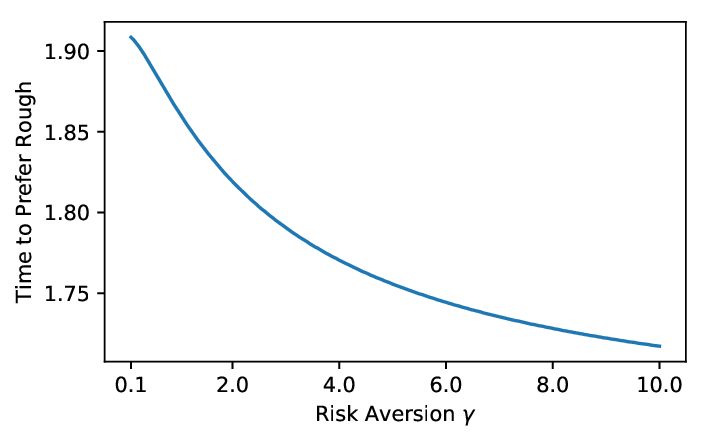}
	\caption{When to prefer rough. The interactions between any two curves in Figure \ref{Fig:ConstMV} are not identical. Thus, we consider $H=0.1$ and $H=0.5$ only. $\gamma \in [0.1, 10]$. The other parameters are the same as in Figure \ref{Fig:ConstMV}.}\label{Fig:PreRough}
\end{figure}
\subsubsection{Comparison with pre-committed MVP and CRRA utility}
We first summarize the comparison of parameter dependence for the const-MV, log-MV, pre-committed MVP, and CRRA utility cases. All four strategies depend on the risk premium $\theta$, risk aversion\footnote{Pre-committed MVP is replaced by target terminal wealth.} $\gamma$, the correlation $\rho$ between stock and volatility, volatility-of-volatility $\sigma$, the mean reversion speed $\kappa$, the investment horizon $T$, and the Hurst parameter $H$. For simplicity, we refer to these seven parameters as the {\it primary parameters}. The CRRA and log-MV strategies have the simplest parameter dependence. Only the primary parameters affect stock demand. The const-MV strategy relies on an additional parameter, namely the risk-free rate. The pre-committed MVP strategy depends on all parameters. Interestingly, the pre-committed MVP strategy is the only strategy that depends on the long-term mean level of volatility. Table \ref{Tab:Dep} summarizes the comparison along with wealth dependence.

\cite{han2019mean} observes that volatility-of-volatility has a material impact on the pre-committed MVP in the sensitivity analysis. However, it is not observed for const-MV and log-MV investors. Alternately, the investment horizon becomes vital for the time-consistent alternatives. 

\cite{dai2020log} argues that log-MV and CRRA criteria are analogous. Their structure of strategies is almost the same. This similarity is maintained when volatility is rough, but the effect of roughness is disparate. In short investment horizons (e.g., 1 year), CRRA investors allocate less to the stock when volatility is rougher \citep{han2019merton}. Log-MV investors do the opposite. In long investment horizons (e.g., 10 years), CRRA and log-MV investors still have distinct preferences regarding roughness. 

\begin{table}
	\small
	\centering
	\begin{tabular}{c | c | c }
		\hline
		Type & Parameter Dependence	 &   Wealth Dependence \\ 
		\hline 
		CRRA	& Primary parameters & Proportional  \\ 
		\hline
		Log-MV & Primary parameters &  Proportional \\ 
		\hline
		Const-MV  &  Primary parameters, risk-free rate $\varUpsilon$  & No\\
		\hline
		Pre-committed & All parameters: primary parameters, initial wealth $M_0$, &  Linear \\
		MVP & initial volatility $\nu_0$, risk-free rate $\varUpsilon$, volatility mean level $\phi$ &  \\
		\hline
	\end{tabular}
	\caption{Summary of the parameter and wealth dependence of the strategies.}
	\label{Tab:Dep}
\end{table}

\subsubsection{Linear controlled Volterra processes} 

\begin{figure}
	\centering
	\includegraphics[width=0.4\textwidth]{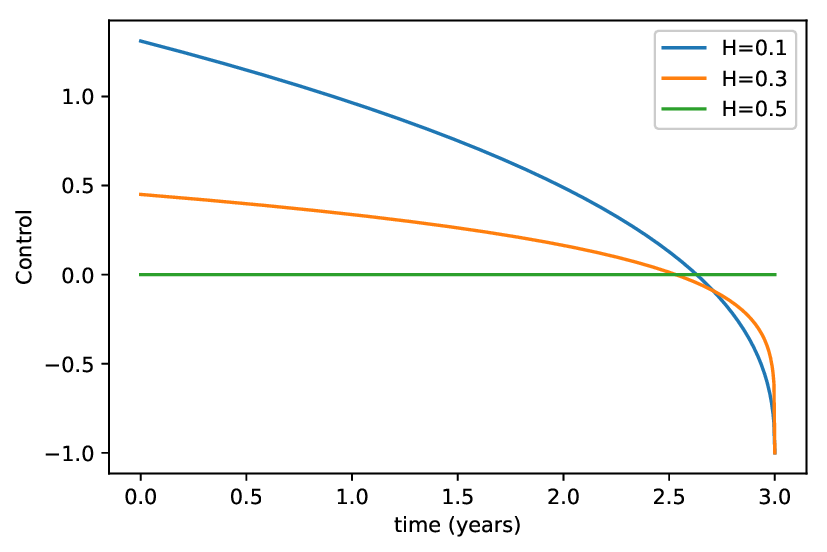}
	\caption{Equilibrium control under the state process \eqref{Eq:LC-MVstate}. We set $B, D, \sigma = 1$, risk aversion $\gamma=1$, and time horizon $T=3$. }\label{Fig:LC}
\end{figure}

To further enrich the analysis, we include a numerical study on the objective \eqref{Eq:LC-MVreward} under the state process \eqref{Eq:LC-MVstate}. For simplicity, consider constant parameters $B, D, \sigma$ and the fractional kernel. The equilibrium control \eqref{Eq:LC-u} then reduces to
\begin{equation*}
\hat \bu_t = \frac{B}{ \gamma D^2} (T - t)^{\frac{1}{2} - H} \Gamma(H+\frac{1}{2}) - \frac{\sigma}{D}.
\end{equation*}
When $H = 1/2$,  $\hat \bu_t$ is a constant. If the state process is rough with $H<1/2$, then $\hat \bu_t$ is larger at the beginning and smaller at the end, as shown in Figure \ref{Fig:LC}. Moreover, this phenomenon is the same for different time horizons. As the fractional kernel appears together with the control, it exhibits distinct effects compared with the rough Heston model. When the agent is more risk-averse (a larger $\gamma$), $\hat \bu_t$ is smaller. In the extreme case with $\gamma \rightarrow \infty$, the effect of roughness is eliminated and $\hat \bu_t \to - \sigma/D$.

\subsection{Simulation study}
In the sensitivity analysis, we only vary the roughness level. In general, this is unrealistic as other parameters are also likely to change. For example, \cite{abi2019lifting,abi2019multifactor} document the connection between volatility roughness and components with fast mean reversion. It is also observed that to capture the deep near-term volatility skew, calibration with the classic Heston model usually results in a greater mean reversion speed $\kappa$ and in a greater volatility-of-volatility $\sigma$.

In this section, we leverage the information from implied volatility (IV). Theoretically, the roughness estimates based on realized and implied volatility should coincide. However, this does not hold in reality. The IV surface represents the current view of the future. In particular, ATM skew explosion indicates near-term downside risks, such as earnings announcements \citep{glasserman2019buy}. As the real parameters underlying the IV surface are unknown, we first adopt the simulated IV surface in \cite{abi2019lifting}. Given the simulated IV data, investors calibrate two sets of parameters for the Heston and rough Heston models. We contrast the two strategies induced from the calibrated parameters. The investor under the Heston model (hereafter, Heston investor) uses the calibrated parameters in \citet[Table 6]{abi2019lifting}. The investor under the rough Heston model (hereafter, rough investor) uses \citet[Table 4]{abi2019lifting}. The variance process is simulated using the lifted Heston method in \cite{abi2019lifting}. The parameters for simulation are given in \citet[Equations (23) and (26)]{abi2019lifting}. We set $M_0 = 1$, $\varUpsilon = 0.01$, $\theta = 1.5$, $T = 10$, and $\gamma = 0.5$. Additionally, we implement the Euler scheme for the stock process. The simulation is run with $250$ time steps for 1 year, corresponding to the $250$ trading days per year.  

\begin{figure}
	\centering
	\begin{minipage}{0.5\textwidth}
		\centering
		\includegraphics[width=0.9\textwidth]{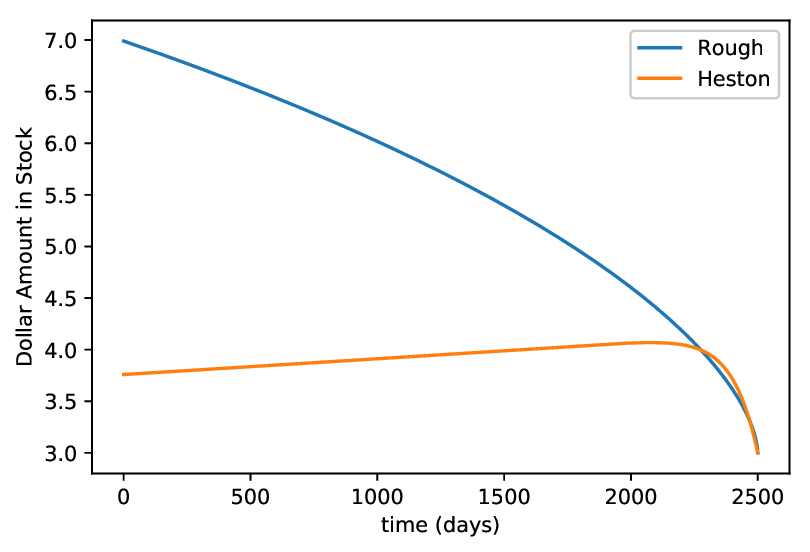}
		\subcaption{Const-MV}\label{Fig:ConstMVSimDollar}
	\end{minipage}%
	\begin{minipage}{0.5\textwidth}
		\centering
		\includegraphics[width=0.9\textwidth]{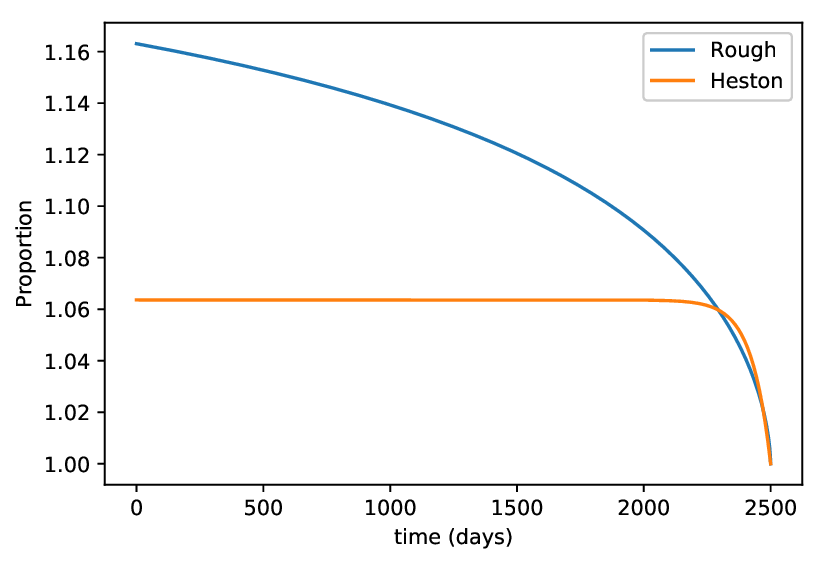}
		\subcaption{Log-MV}\label{Fig:LogMVSimPo}
	\end{minipage}
	\caption{Dollar amount of wealth for the const-MV strategy and proportional amount of wealth for the log-MV strategy. In both cases, roughness has a significant impact.}\label{Fig:ImInvest}
\end{figure}

Figure (\ref{Fig:ConstMVSimDollar}) plots the dollar amount in the stock for the const-MV investor, and Figure (\ref{Fig:LogMVSimPo}) depicts the proportional amount of wealth in the stock for the log-MV investor. In both cases, the investors demand more if the rough Heston model is adopted. Volatility roughness dramatically changes the investment decisions. The rough const-MV investor almost doubles the Heston counterpart. The rough log-MV investor increases stock demand by nearly 10\%. These results are interesting. The advantage of the rough Heston model is that it better captures the volatility smiles with short maturities. In other words, the rough Heston model captures the near-term downside risk, whereas the classic Heston model fails to do so. Surprisingly, the downside risk alters the investment over almost the entire time horizon. We interpret these adjustments as hedging against risk.

\begin{figure}
	\centering
	\begin{minipage}{0.5\textwidth}
		\centering
		\includegraphics[width=0.9\textwidth]{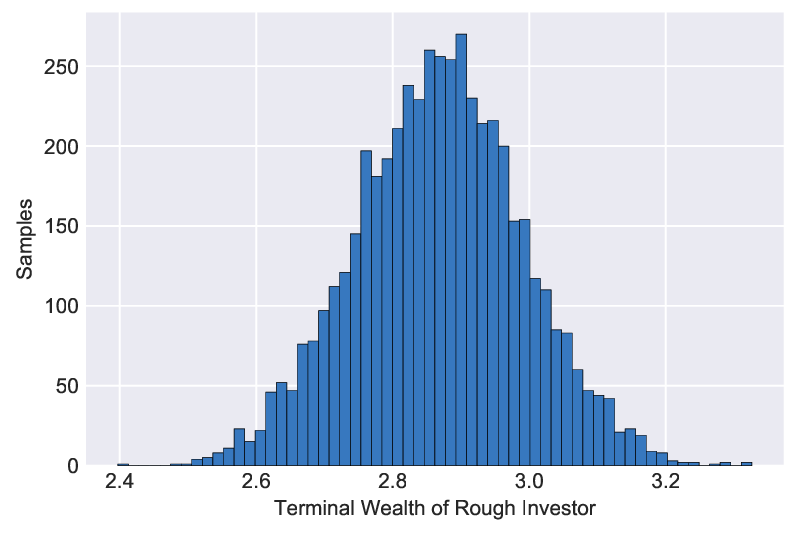}
		\subcaption{}\label{Fig:ConstMVX}
	\end{minipage}%
	\begin{minipage}{0.5\textwidth}
		\centering
		\includegraphics[width=0.9\textwidth]{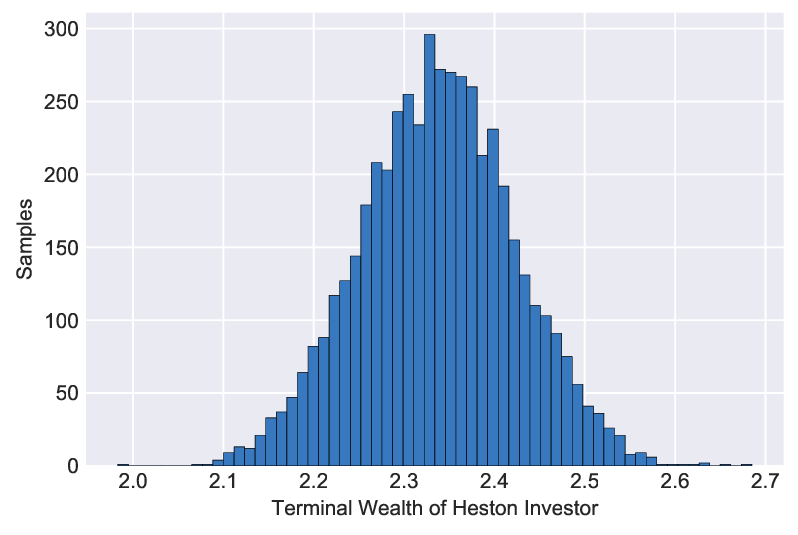}
		\subcaption{}\label{Fig:ConstMVXHn}
	\end{minipage}
	\caption{Distribution of terminal wealth for the const-MV case. The rough investor achieves terminal wealth with a sample mean of 2.868 and a sample variance of 0.015. The Heston investor has a sample mean of 2.337 and a variance of 0.007.}\label{Fig:ConstMVTer}
\end{figure}

\begin{figure}
	\centering
	\begin{minipage}{0.5\textwidth}
		\centering
		\includegraphics[width=0.9\textwidth]{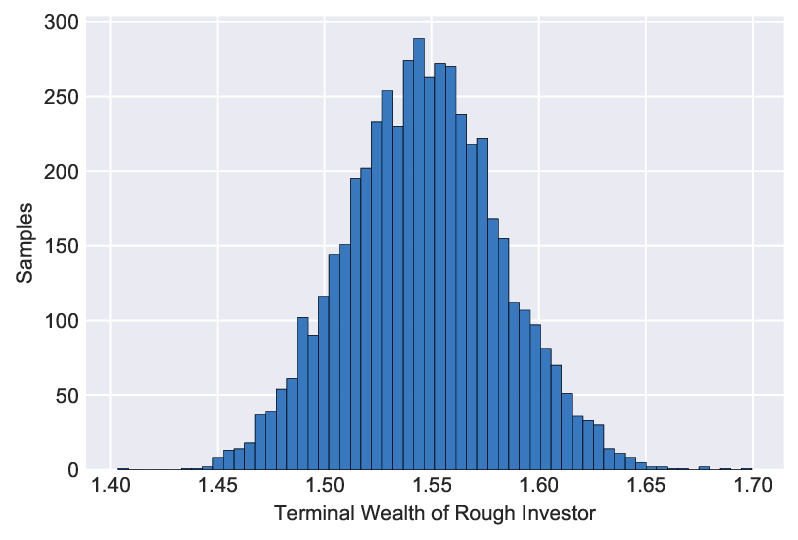}
		\subcaption{}\label{Fig:LogMVX}
	\end{minipage}%
	\begin{minipage}{0.5\textwidth}
		\centering
		\includegraphics[width=0.9\textwidth]{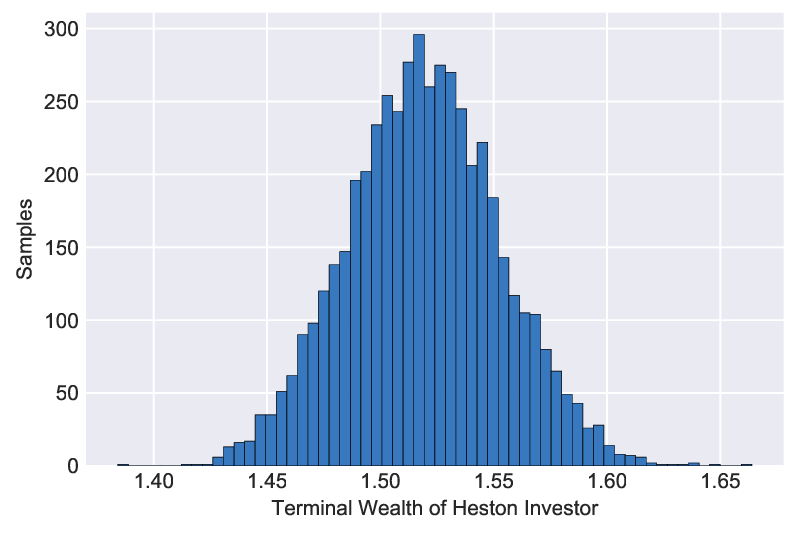}
		\subcaption{}\label{Fig:LogMVXHn}
	\end{minipage}
	\caption{Distribution of terminal wealth for the log-MV case. The rough investor achieves terminal wealth with a sample mean of 1.546 and a sample variance of 0.0013. The Heston investor has a sample mean of 1.519 and a sample variance of 0.0011.}\label{Fig:LogMVTer}
\end{figure}

Figures \ref{Fig:ConstMVTer} and \ref{Fig:LogMVTer} demonstrate the distribution of terminal wealth for the const-MV and log-MV investors with the rough Heston model or the classic Heston model. Figures \ref{Fig:ConstMVTer} and \ref{Fig:LogMVTer} have 5,000 simulation paths. Rough investors tend to obtain higher terminal wealth but with higher variance. The reason for this may be that they bear the near-term risk represented by roughness. The log-MV case in Figure \ref{Fig:LogMVTer} looks conservative compared with the const-MV case. However, the const-MV case in Figure \ref{Fig:ConstMVTer} and the log-MV case in Figure \ref{Fig:LogMVTer} are not fairly comparable. Risk aversion $\gamma$ is set to $0.5$, but one is for wealth and another is for log return. The scale makes the actual risk aversion level different.

\subsection{Empirical study}\label{Sec:Empirical}
To further elucidate the need to account for rough volatility, we evaluate the performance of the strategies on real financial data. We download the CBOE S\&P 500 options data from OptionMetrics via Wharton Research Data Services. The most recent 1-year time period we can obtain is from 2018/05/01 to 2019/05/01. The U.S. equity market was relatively volatile during this period, plunging in October 2018 and December 2018. Due to the computational burden involved in calibrating the rough Heston model, we adopt the following procedure to implement the trading strategies. Consider the investment horizon as 1 month. The models are calibrated on European call options data observed on the first trading day of that month. We then use the calibrated parameters for trading in the whole month. The portfolios are rebalanced at a daily frequency (i.e., the step size equals $1/250$). For simplicity, we ignore transaction costs and price impact. After this month, we redo the calibration and implement the trading rules for the next 1-month investment horizon. Overall, for each model, we have 12 sets of calibrated parameters in a year. Tables \ref{Tab:CaliRough} and \ref{Tab:CaliHeston} report the sample means and sample standard deviations (std) of the calibration results. In Table \ref{Tab:CaliRough}, the calibrated Hurst parameter is around 0.15 with a slight variance. Therefore, volatility roughness is a persistent fact in the equity market. In Table \ref{Tab:CaliHeston}, it is further confirmed that the classic Heston model tends to yield a considerably greater mean reversion speed $\kappa$ and volatility-of-volatility $\sigma$.

\begin{table}[!htbp]
	\centering
	\begin{tabular}{c | c | c | c | c }
		\hline
		Parameter & $H$   &   $\kappa$    &      $\rho$     &        $\sigma$           \\
		\hline
		Mean(std) & 0.1576 (0.0170)  &  0.2378 (0.0817) & -0.7611 (0.0648) & 0.4259 (0.1173) \\
		\hline
	\end{tabular}
	\caption{Calibrated values for the rough Heston model. We clean the options data using a procedure similar to that used in \cite{glasserman2019buy}. Only contracts with reasonable open interest and implied volatility are included. We restrict the maturity to a maximum of 1 year. For simplicity, we only report the values for the parameters that appear in the strategies.}
	\label{Tab:CaliRough}
\end{table}

\begin{table}[!htbp]
	\centering
	\begin{tabular}{c | c | c | c }
		\hline
		Parameter &   $\kappa$    &      $\rho$     &        $\sigma$           \\
		\hline
		Mean(std) &  3.6887 (1.2201) &  -0.6095  (0.0768) & 1.9928  (1.9930) \\
		\hline
	\end{tabular}
	\caption{Calibrated values for the classic Heston model. The same options dataset from Table \ref{Tab:CaliRough} is used. Larger $\kappa$ and $\sigma$ values are confirmed.}
	\label{Tab:CaliHeston}
\end{table}

In Tables \ref{Tab:SharpeConst} and \ref{Tab:SharpeLog}, we present the terminal wealth and Sharpe ratio at several risk aversion levels. In the const-MV and log-MV cases, trading strategies under rough volatility dominate the classic Heston counterparts \citep{basak2010dynamic,dai2020log} with both higher terminal wealth and a better Sharpe ratio. This confirms that the rough Heston model captures profit opportunities missed by the classic Heston model. Rough volatility models are therefore an attractive alternative for volatility modeling. Furthermore, trading strategies under the log-MV criterion are less sensitive to the choice of risk aversion. Figure \ref{Fig:SP} illustrates the special case in which risk aversion equals $2$. Const-MV strategies yield higher profits, whereas log-MV strategies are more stable, even during the market downturn in October and December 2018. Log-MV strategies yield lower Sharpe ratios mainly due to lower excess returns. In addition, we emphasize that directly comparing the const-MV and log-MV cases is not fair as the actual risk aversion preferences are distinct.

\begin{table}[!htbp]
	\small
	\centering
	\begin{tabular}{c | c | c | c | c | c  }
		\hline
		$\gamma$   &   0.5    &    1      &       2    &    5   &  10      \\
		\hline
		Rough  & 1.7965 (1.0503)   &   1.4084 (1.0488)  &  1.2143 (1.0463)   & 1.0978 (1.0442)   & 1.0590 (1.0441) \\
		\hline
		Heston & 1.5733 (0.8682)  &  1.2968 (0.8071)    &  1.1585 (0.7872)   & 1.0755 (0.7755)    & 1.0479 (0.7722) \\
		\hline
	\end{tabular}
	\caption{Performance of the const-MV strategy, reported with terminal wealth (annualized Sharpe ratio in brackets). We set initial wealth $M_0 = 1$ and risk-free rate $\varUpsilon = 0.02$. The risk premium parameter $\theta$ is estimated using a rolling window method, with a window size of 3 years. ``Rough" indicates the strategy under the rough Heston model and ``Heston" indicates that under the classic Heston model.}
	\label{Tab:SharpeConst}
\end{table}

\begin{table}[!htbp]
	\small
	\centering
	\begin{tabular}{c | c | c | c | c | c  }
		\hline
		$\gamma$   &   0.5    &    1      &       2    &    5    & 10        \\
		\hline
		Rough & 1.2120 (0.8945) &  1.1739 (0.9199)  & 1.1303 (0.9517)  &  1.0798 (0.9920) & 1.0540 (1.0145) \\
		\hline
		Heston & 1.1971 (0.8415) & 1.1550 (0.8291) & 1.1113 (0.8139)   &  1.0664 (0.7954)   & 1.0455 (0.7864) \\
		\hline
	\end{tabular}
	\caption{Performance of the log-MV strategies. The same parameters from Table \ref{Tab:SharpeConst} are adopted.}
	\label{Tab:SharpeLog}
\end{table}

\begin{figure}[!htbp]
	\centering
	\begin{minipage}{0.5\textwidth}
		\centering
		\includegraphics[width=0.99\textwidth]{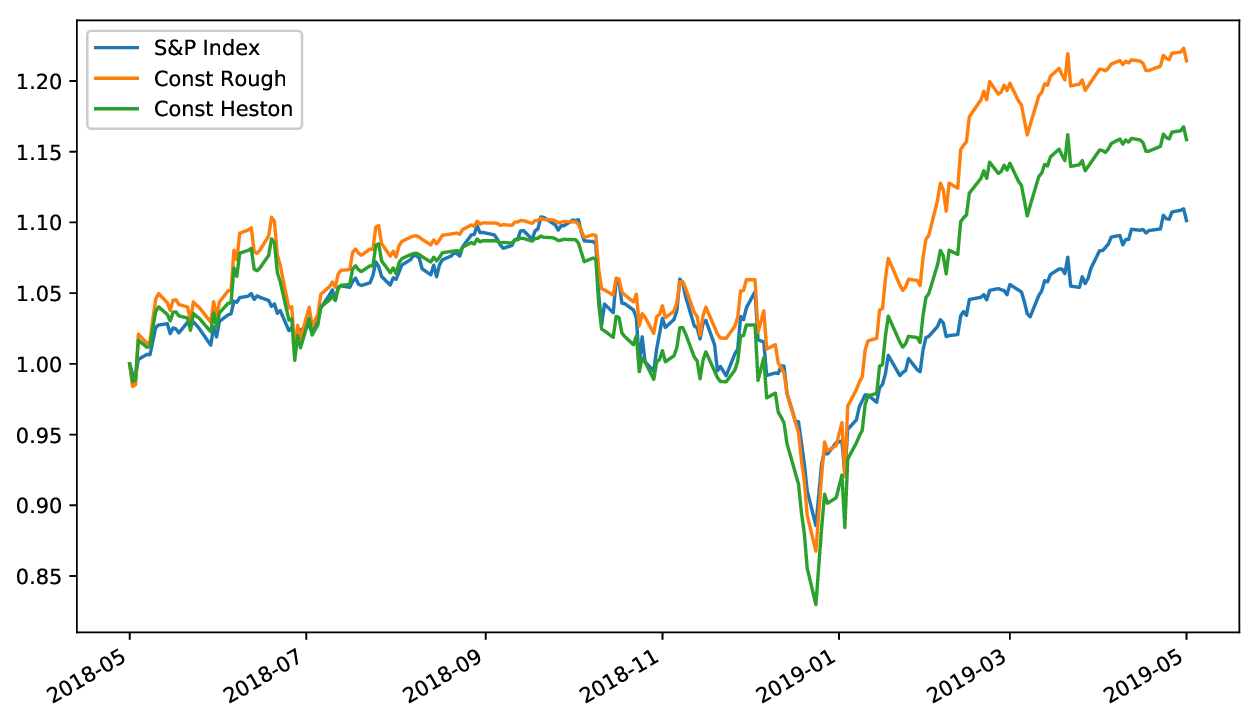}
		\subcaption{Const-MV}\label{Fig:ConstSP}
	\end{minipage}%
	\begin{minipage}{0.5\textwidth}
		\centering
		\includegraphics[width=0.99\textwidth]{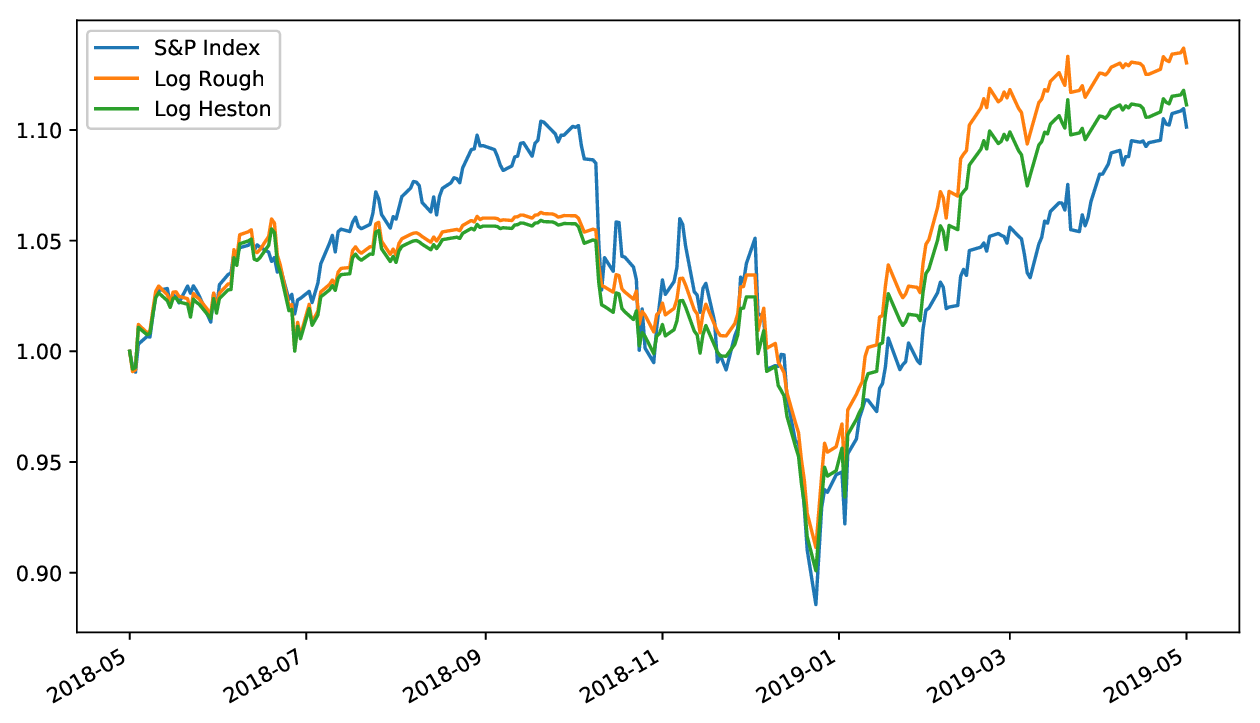}
		\subcaption{Log-MV}	\label{Fig:LogSP}
	\end{minipage}
	\caption{Wealth paths under the const-MV and log-MV criteria, with risk aversion $\gamma = 2$. The other parameters are set as in Table \ref{Tab:SharpeConst}. The S\&P 500 Index is normalized such that the initial value is 1. The trading strategy using the rough Heston model achieves higher terminal wealth with a better Sharpe ratio, as detailed in Tables \ref{Tab:SharpeConst} and \ref{Tab:SharpeLog}. Strategies under the log-MV criterion are more stable than the const-MV case.}\label{Fig:SP}
\end{figure}

\section{Concluding remarks}\label{Sec:Con}
In this paper, we use the functional It\^o calculus to examine in depth the equilibrium strategies under time-inconsistent preferences and in a rough stochastic environment. Volatility roughness significantly modifies investment demands. Our general framework also embraces Volterra processes with potential applications beyond rough volatility. Several interesting problems are left for future research. The first is the existence and uniqueness of solutions to the extended PHJB equation system in Definition \ref{Def:EP-HJB}. The second is the time-inconsistent open-loop control problem under Volterra processes.

\appendix
\section{Brief summary of the functional It\^o calculus (Viens and Zhang, 2019)}\label{Sec:FuncIto}

Let $\Omega \triangleq C^0([0, T],\; \R^n)$ be the sample space with continuous paths, $\bar \Omega  \triangleq D^0([0, T], \; \R^n)$ be the sample space with c\`adl\`ag (right-continuous with left limits) paths, and
\begin{align*}
& \Omega_t  \triangleq C^0([t, T],\; \R^n), \quad \Lambda \triangleq [0, T] \times \Omega, \quad \bar \Lambda \triangleq \left\{ (t, \omega) \in [0, T] \times \bar \Omega : \omega\big|_{[t, T]} \in \Omega_t \right\}, \\
& ||\omega||_T \triangleq \sup_{ 0 \leq t \leq T} |\omega_t|, \quad \bd((t, \omega), (t', \omega')) \triangleq |t - t'| + ||\omega - \omega'||_T.
\end{align*}

Let $C^0(\bar \Lambda)$ be the space of functions $f: \bar \Lambda \rightarrow \R$, which are continuous under $\bd$. For $f \in C^0(\bar \Lambda)$, define the time derivative as follows:
\begin{equation}\label{Eq:t_deri}
\partial_t f(t, \omega) \triangleq \lim_{\delta \downarrow 0} \frac{f(t + \delta, \omega) - f(t, \omega)}{\delta}, \quad \text{ for all } (t, \omega) \in \bar \Lambda,
\end{equation}
whenever the limit exists.

Given $(t, \omega) \in \bar \Lambda$, the spatial derivative with respect to $\omega$, denoted by $\partial_\omega f(t, \omega)$, is a linear operator on $\Omega_t$ and defined as the Fr\'echet derivative with respect to $\omega \id_{[t, T]}$. That is,
\begin{equation}\label{Eq:spatial1_deri}
f(t, \omega + \eta \id_{[t, T]}) - f(t, \omega) = \ang{\partial_\omega f(t, \omega), \eta} + o(||\eta\id_{[t, T]}||_T), \quad \text{ for any } \eta \in \Omega_t.
\end{equation}
Furthermore, $\partial_\omega f(t, \omega)$ satisfies the definition of the Gateaux derivative:
\begin{equation}
\ang{\partial_\omega f(t, \omega), \eta} = \lim_{\varepsilon \rightarrow 0} \frac{f(t, \omega + \varepsilon \eta \id_{[t, T]}) - f(t, \omega)}{\varepsilon}, \quad \text{ for any } \eta \in \Omega_t.
\end{equation}
The perturbation is on $[t, T]$, not on $[0, t)$. If $\eta \in \Omega_s$ for certain $s < t$, the derivative is understood as follows:
\begin{equation}\label{Eq:Convention}
\ang{\partial_\omega f(t, \omega), \eta} \triangleq \ang{\partial_\omega f(t, \omega), \eta\id_{[t, T]}}. 
\end{equation}

The second-order derivative $\partial^2_{\omega \omega} f(t, \omega)$ is defined as a bilinear operator on $\Omega_t \times \Omega_t$:
\begin{equation}\label{Eq:spatial2_deri}
\ang{\partial_\omega f(t, \omega + \eta_1 \id_{[t, T]}), \eta_2}  - \ang{\partial_\omega f(t, \omega), \eta_2} = \ang{\partial^2_{\omega\omega} f(t, \omega), (\eta_1, \eta_2)} + o(||\eta_1 \id_{[t, T]}||_T),
\end{equation} 
for any $\eta_1, \eta_2 \in \Omega_t$. If $\eta_1, \eta_2 \in \Omega_s$ for certain $s < t$, the derivative is understood in the same way as in (\ref{Eq:Convention}).

For the well-posedness of these derivatives, we refer readers to \citet[Proposition 3.7]{viens2017martingale}.

We introduce two spaces, namely $C^{1, 2}_+ (\Lambda)$ and $C^{1, 2}_{+, \alpha} (\Lambda)$, from \cite{viens2017martingale}, under which the functional It\^o formula in \citet[Theorem 3.10 and 3.17]{viens2017martingale} holds.
\begin{definition}[{\citet[Definition 3.3]{viens2017martingale}}]
	Suppose that $f \in C^0(\bar \Lambda)$ and $\partial_\omega f$ exists for all $(t, \omega) \in \bar \Lambda$. 
	\begin{enumerate}[label={(\arabic*).}]
		\item $\partial_\omega f$ is said to have polynomial growth if there exist constants $C_0, m > 0$ such that
		\begin{equation}
		\big|\ang{\partial_\omega f(t, \omega), \eta} \big| \leq C_0 [1 + ||\omega||^m_T] ||\eta\id_{[t, T]}||_T, \quad \forall \; (t, \omega) \in \bar \Lambda, \; \eta \in \Omega.
		\end{equation} 
		\item $\partial_\omega f$ is said to be continuous if, for all $\eta \in \Omega$, the mapping $(t, \omega) \in \bar \Lambda \mapsto \ang{\partial_\omega f(t, \omega), \eta}$ is continuous under $\bd$.
		\item $\partial^2_{\omega\omega} f$ is said to have polynomial growth if there exist constants $C_0, m > 0$ such that
		\begin{align}
		\big|\ang{\partial^2_{\omega\omega} & f(t, \omega), (\eta_1, \eta_2)} \big| \nonumber \\
		& \leq C_0 [1 + ||\omega||^m_T] ||\eta_1\id_{[t, T]}||_T ||\eta_2 \id_{[t, T]}||_T, \; \forall \; (t, \omega) \in \bar \Lambda, \; \eta_1, \eta_2 \in \Omega.
		\end{align} 
		\item $\partial^2_{\omega\omega} f$ is said to be continuous if, for all $\eta_1, \eta_2 \in \Omega$, the mapping $(t, \eta_1, \eta_2) \in \bar \Lambda_2 \mapsto \ang{\partial^2_{\omega\omega} f(t, \omega),$ $ (\eta_1, \eta_2)}$ is continuous under $\bd'((t, \omega_1, \omega_2), (t', \omega'_1, \omega'_2)) \triangleq |t - t'| + ||\omega_1 - \omega'_1||_T + ||\omega_2 - \omega'_2||_T$, where $\bar \Lambda_2 \triangleq \big\{ (t, \omega_1, \omega_2) \in [0, T] \times \bar \Omega \times \bar \Omega : \omega_1\big|_{[t, T]}, \; \omega_2 \big|_{[t, T]}  \in \Omega_t \big\}$.
	\end{enumerate}
\end{definition}
\begin{definition}[{\citet[Definition 3.4]{viens2017martingale}}]
	Let $C^{1, 2}(\bar \Lambda) \subset C^0(\bar \Lambda)$ be the set of all $f$ with continuous derivatives $\partial_t f$, $\partial_\omega f$, $\partial^2_{\omega \omega} f$ on $\bar \Lambda$. Let $C^{1, 2}_+(\bar \Lambda)$ be the set of all $f \in C^{1, 2}(\bar \Lambda)$ such that all derivatives have polynomial growth and $\ang{\partial^2_{\omega\omega} f(t, \omega), (\eta, \eta)}$ is locally uniformly continuous in $\omega$ with polynomial growth, namely there exists a constant $m > 0$ and a bounded modulus of continuity function $\varrho$. For all $(t, \omega), (t, \omega') \in \bar \Lambda$ and $\eta \in \Omega_t$, we have
	\begin{align}
	\big|\ang{\partial^2_{\omega\omega} f(t, \omega) - \partial^2_{\omega\omega} f(t, \omega'), (\eta, \eta)} \big| \leq \big[1 + ||\omega||^m_T + ||\omega'||^m_T \big] || \eta \id_{[t, T]} ||^2_T \varrho(||\omega - \omega'||_T).
	\end{align}
	$C^{1, 2}_+(\Lambda)$ is defined in the same spirit of $C^{1, 2}_+(\bar \Lambda)$, with $\bar \Lambda$ replaced by $\Lambda$.
\end{definition}
\begin{definition}[{\citet[Definition 3.16]{viens2017martingale}}]
	$f \in C^{1, 2}_+ (\Lambda)$ is said to {\bf vanish diagonally at a rate of $\alpha \in (0, 1)$}, denoted by $f \in C^{1, 2}_{+, \alpha} (\Lambda)$, if there exists an extension of $f$ in $C^{1, 2}_+ (\bar \Lambda)$, still denoted by $f$, such that for every $ 0 \leq t < T$, $0 < \delta \leq T -t$, and $\eta, \eta_1, \eta_2 \in \Omega_t$ with supports contained in $[t, t + \delta]$:
	\begin{enumerate}[label={(\arabic*).}]
		\item $\forall \; \omega \in \bar \Omega$ satisfying $\omega\id_{[t, T]} \in \Omega_t$,
		\begin{align}
		\big|\ang{\partial_\omega f(t, \omega), \eta} \big| & \leq C_0 [1 + ||\omega||^m_T] ||\eta||_T \delta^\alpha, \\
		\big|\ang{\partial^2_{\omega\omega} f(t, \omega), (\eta_1, \eta_2)} \big| & \leq C_0 [1 + ||\omega||^m_T] \big|\big||\eta_1||\eta_2|\big|\big|_T \delta^{2\alpha};
		\end{align}
		\item For any other $\omega' \in \bar \Omega$ satisfying $\omega'\id_{[t, T]} \in \Omega_t$,
		\begin{align}
		\big|\ang{\partial_\omega f(t, \omega) & - \partial_\omega f(t, \omega'), \eta} \big| \nonumber \\ 
		&\leq \big[1 + ||\omega||^m_T + ||\omega'||^m_T \big] || \eta||_T \varrho(||\omega - \omega'||_T) \delta^\alpha, \\
		\big|\ang{\partial^2_{\omega\omega} f(t, \omega) & - \partial^2_{\omega\omega} f(t, \omega'), (\eta_1, \eta_2)} \big| \nonumber \\
		&\leq \big[1 + ||\omega||^m_T + ||\omega'||^m_T \big] \big|\big||\eta_1||\eta_2|\big|\big|_T \varrho(||\omega - \omega'||_T) \delta^{2\alpha}.
		\end{align}
	\end{enumerate}
	Constant $m > 0$ denotes the polynomial growth rate and $\varrho$ is a bounded modulus of the continuity function.
\end{definition}
$\alpha$ characterizes the level of singularity in the diagonal of time. Finally, the functional It\^o formula is quoted in Theorem \ref{Thm:Ito}.

\section{Proofs of results}\label{Sec:Proofs}
\subsection{Proof of Lemma \ref{Lem:Recursion}}
\begin{proof}
	By the tower property of the conditional expectation and the definitions of $c^{r, \bu}$, $f^\bu$, $g^\bu$, and $\omega^{t+h}$,  
	\begin{align}\label{Eq:J(t)}
	J(t, \omega; \bu) = &  \int^T_t \E\Big[ \E\Big[C(t, \omega_t, r,  X^{t, \omega, \bu}_\rdot, \bu(r, X^{t, \omega, \bu}_\rdot)) \Big| \cF_{t+h} \Big] \Big| \cF_t \Big] dr  \cr
	& + \E\Big[ \E\Big[ F(t, \omega_t, X^{t, \omega, \bu}_{T \wedge \cdot}) \Big| \cF_{t+h} \Big] \Big| \cF_t \Big] + G\big(t, \omega_t, \E\big[ \E\big[ X^{t, \omega, \bu}_T \big| \cF_{t+h} \big] \big| \cF_t\big] \big) \cr
	= & \int^T_t \E\Big[ c^{r, \bu}(t+h, \omega^{t + h}, t, \omega_t)\Big| \cF_t \Big] dr  \\
	& + \E\Big[ f^\bu(t+h, \omega^{t+h}, t, \omega_t) \Big| \cF_t \Big] + G\big(t, \omega_t, \E\big[ g^\bu(t+h, \omega^{t+h}) \big| \cF_t\big] \big). \nonumber
	\end{align}
	Meanwhile, the definition of the reward functional in \eqref{Eq:Reward} indicates the following:
	\begin{align}\label{Eq:J(t+h)}
	J(t+h, \omega^{t+h}; \bu) =& \E\Big[ \int^T_{t+h} C(t+h, \omega^{t+h}_{t+h}, r,  X^{t+h, \omega^{t+h}, \bu}_\rdot, \bu(r, X^{t+h, \omega^{t+h}, \bu}_\rdot)) dr  \Big| \cF_{t+h} \Big] \cr
	&+ \E\big[ F(t+h, \omega^{t+h}_{t+h}, X^{t+h, \omega^{t+h}, \bu}_{T \wedge \cdot}) \big| \cF_{t+h} \big] \cr
	& + G(t+h, \omega^{t+h}_{t+h}, \E[ X^{t+h, \omega^{t+h}, \bu}_T | \cF_{t+h}]) \cr
	=& \int^T_{t+h} c^{r, \bu}(t+h, \omega^{t+h}, t+h, \omega^{t+h}_{t+h}) dr + f^\bu(t+h, \omega^{t+h}, t+h, \omega^{t+h}_{t+h}) \cr
	& + G(t+h, \omega^{t+h}_{t+h}, g^\bu(t+h, \omega^{t+h})), 
	\end{align}
	where
	\begin{eqnarray}
	X^{t+h, \omega^{t+h}, \bu}_s &=& \omega^{t+h}_s + \int^s_{t+h} \mu(s; r, X^{t+h, \omega^{t+h}, \bu}_\rdot, \bu(r, X^{t+h, \omega^{t+h}, \bu}_\rdot)) dr \cr
	& & + \int^s_{t+h} \sigma(s; r, X^{t+h, \omega^{t+h}, \bu}_\rdot, \bu(r, X^{t+h, \omega^{t+h}, \bu}_\rdot)) dW_r, \; t+h \leq s \leq T, \cr
	X^{t+h, \omega^{t+h}, \bu}_s &=& \omega^{t+h}_s, \; 0 \leq s < t+h.
	\end{eqnarray}
	
	Taking the conditional expectation at $\cF_t$ on both sides of (\ref{Eq:J(t+h)}) yields the following:
	\begin{align}\label{Eq:E[J(t+h)]}
	\E \Big[ J(t+h, \omega^{t+h}; \bu) \Big| \cF_t \Big] =& \int^T_{t+h} \E \Big[ c^{r, \bu}(t+h, \omega^{t+h}, t+h, \omega^{t+h}_{t+h}) \Big| \cF_t \Big] dr \nonumber \\
	& + \E \Big[ f^\bu(t+h, \omega^{t+h}, t+h, \omega^{t+h}_{t+h}) \Big| \cF_t \Big] \nonumber \\
	& + \E \Big[ G(t+h, \omega^{t+h}_{t+h}, g^\bu(t+h, \omega^{t+h})) \Big| \cF_t \Big].
	\end{align}
	The result follows by combining (\ref{Eq:J(t)}) and (\ref{Eq:E[J(t+h)]}).
\end{proof}
\subsection{Proof of Lemma \ref{Lem:TrueMartingale}}
\begin{proof}
	Let $m$ be a generic positive value that may vary from line to line. We first present the proof for the regular case. By (1)-(2) in Definition \ref{Def:U} and the assumption that $\partial_\omega f$ has polynomial growth,
	\begin{align}
	& \E \Big[\int^T_t \big| \ang{\partial_\omega f(r, X^\bu \otimes_r \Theta^{r, \bu}), \sigma^{r, \bu}} \big|^2 dr \Big| \cF_t \Big] \\
	& \leq C_0 \E \Big[ \sup_{ t \leq r \leq T} \big| \ang{\partial_\omega f(r, X^\bu \otimes_r \Theta^{r, \bu}), \sigma^{r, \bu}} \big|^2 \Big| \cF_t \Big] \cr
	& \leq C_0 \E \Big[\Big( 1 + \sup_{ t \leq r \leq T} \sup_{0 \leq s \leq T} |(X^\bu \otimes_r \Theta^{r, \bu})_s|^m \Big)^2  \sup_{ t \leq r \leq T} \sup_{ r \leq s \leq T} | \sigma^\bu(s; r, X^\bu_\rdot)|^2 \Big| \cF_t \Big] \cr
	& \leq C_0 \E \Big[ \Big(1 + \sup_{ t \leq r \leq T} \sup_{0 \leq s \leq T} |(X^\bu \otimes_r \Theta^{r, \bu})_s|^m \Big) \Big| \cF_t \Big] \cr
	& \leq C_0 \E \Big[ \Big(1 + \sup_{ 0 \leq s \leq T} |X^\bu_s|^p \Big)\Big| \cF_t \Big] < \infty. \nonumber 
	\end{align}
	For the singular case, for $[r, T]$, consider the partition $r = r_\infty < ... < r_k < ... < r_0 = T$, where $r_k = r + \frac{T-r}{2^k}$. Then,
	\begin{align}
	& \E \Big[\int^T_t \big| \ang{\partial_\omega f(r, X^\bu \otimes_r \Theta^{r, \bu}), \sigma^{r, \bu}} \big|^2 dr \Big| \cF_t \Big] \\
	& = \E \Big[\int^T_t \big| \lim_{\delta \downarrow 0} \ang{\partial_\omega f(r, X^\bu \otimes_r \Theta^{r, \bu}), \sigma^{\delta, r, \bu}} \big|^2 dr \Big| \cF_t \Big] \cr
	& = \E \Big[\int^T_t \big| \lim_{\delta \downarrow 0} \sum^\infty_{k = 0} \ang{\partial_\omega f(r, X^\bu \otimes_r \Theta^{r, \bu}), \sigma^{\delta, r, \bu}_s \id_{s \in [r_{k+1}, r_k)}} \big|^2 dr \Big| \cF_t \Big] \cr
	& \leq C_0 \E \Big[\int^T_t \Big(1 + \sup_{0 \leq s \leq T} |(X^\bu \otimes_r \Theta^{r, \bu})_s|^m \Big)^2 \big| \lim_{\delta \downarrow 0} \sum^\infty_{k = 0} ||\sigma^{\delta, r, \bu}_s \id_{s \in [r_{k+1}, r_k)}||_T (r_k - r_{k+1})^\alpha \big|^2 dr \Big| \cF_t \Big] \nonumber,
	\end{align}
	where we use the assumption that $f$ vanishes diagonally at a rate of $\alpha \in (0, 1)$ and the fact that $\partial_\omega f$ is a linear operator in the last inequality. By (2) in Definition \ref{Def:U} for the singular case,
	\begin{align}
	& \sum^\infty_{k = 0} ||\sigma^{\delta, r, \bu}_s \id_{s \in [r_{k+1}, r_k)}||_T (r_k - r_{k+1})^\alpha  \\
	&\leq  C_0 \Big(1 + \sup_{0 \leq s \leq T} |(X^\bu \otimes_r \Theta^{r, \bu})_s|^m \Big) \sum^\infty_{k = 0} \big( r_{k+1} \vee (r+\delta) - r \big)^{H - 1/2} \Big( \frac{T-r}{2^{k+1}} \Big)^\alpha. \nonumber
	\end{align}
	For any $0 < \delta \leq T - r$, there exists an integer $z$ such that $\frac{T-r}{2^{z+1}} < \delta \leq \frac{T-r}{2^z}$. Then, 
	\begin{align}
	& \sum^\infty_{k = 0} \big(r_{k+1} \vee (r+\delta) - r \big)^{H - 1/2} \Big( \frac{T-r}{2^{k+1}} \Big)^\alpha \cr 
	& = \sum^{z-1}_{k=0} \Big(\frac{T-r}{2^{k+1}}\Big)^{H - 1/2} \Big( \frac{T-r}{2^{k+1}} \Big)^\alpha + \sum^\infty_{k = z} \delta^{H - 1/2} \Big( \frac{T-r}{2^{k+1}} \Big)^\alpha \cr 
	& = \sum^{z-1}_{k=0} \Big(\frac{T-r}{2^{k+1}}\Big)^\beta + \delta^{H - 1/2} \Big(\frac{T-r}{2^{z+1}} \Big)^\alpha \sum^\infty_{k = 0}  \Big( \frac{1}{2^\alpha} \Big)^k.
	\end{align}
	Note that $\frac{T-r}{2^{z+1}} < \delta $ implies that $\Big(\frac{T-r}{2^{z+1}} \Big)^\alpha < \delta^\alpha$. We obtain the following:
	\begin{align}
	& \sum^\infty_{k = 0} ||\sigma^{\delta, r, \bu}_s \id_{s \in [r_{k+1}, r_k)}||_T (r_k - r_{k+1})^\alpha \cr 
	& \leq C_0 \Big(1 + \sup_{0 \leq s \leq T} |(X^\bu \otimes_r \Theta^{r, \bu})_s|^m \Big) \big[ (T-r)^\beta + \delta^\beta \big].
	\end{align}
	Finally,
	\begin{align}
	& \E \Big[\int^T_t \big| \ang{\partial_\omega f(r, X^\bu \otimes_r \Theta^{r, \bu}), \sigma^{r, \bu}} \big|^2 dr \Big| \cF_t \Big] \\
	& \leq C_0 \E \Big[\int^T_t \Big(1 + \sup_{0 \leq s \leq T} |(X^\bu \otimes_r \Theta^{r, \bu})_s|^m \Big)^4 (T - r)^{2 \beta} dr \Big| \cF_t \Big] \cr
	& \leq C_0 \E \Big[ \Big(1 + \sup_{ t \leq r \leq T} \sup_{0 \leq s \leq T} |(X^\bu \otimes_r \Theta^{r, \bu})_s|^m \Big) \Big| \cF_t \Big] \cr
	& \leq C_0 \E \Big[ \Big(1 + \sup_{ 0 \leq s \leq T} |X^\bu_s|^p \Big)\Big| \cF_t \Big] < \infty, \nonumber 
	\end{align}
	as desired.
\end{proof}

\subsection{Proof of Theorem \ref{Thm:Verification}}
\begin{proof}
	First, we show that the interpretations in Definition \ref{Def:EP-HJB} (6) hold and $V(t, \omega) = J(t, \omega; \hat \bu)$.
	
	By (\ref{Eq:fsy}), (\ref{Eq:g}), (\ref{Eq:csyr}) and Lemma $\ref{Lem:TrueMartingale}$, $f^{s,y}(t, X^{\hat \bu} \otimes_t \Theta^{t, \hat \bu})$, $g(t, X^{\hat \bu} \otimes_t \Theta^{t, \hat \bu})$, and $c^{s, y, r}(t, X^{\hat \bu} \otimes_t \Theta^{t, \hat \bu})$ are martingales. By the boundary conditions in (\ref{Eq:fsy}), (\ref{Eq:g}), and (\ref{Eq:csyr}) and note that $\omega = X^{t, \omega, \hat \bu} \otimes_t \Theta^{t, \hat \bu}$, we derive the following:
	\begin{align*}
	f^{s, y}(t, \omega) & = \hatE \big[ F(s, y, X^{t, \omega, \hat \bu}_{T \wedge \cdot}) \big| \cF_t \big], \qquad g(t, \omega) = \hatE \big[X^{t, \omega, \hat \bu}_T \big| \cF_t \big], \\
	c^{s, y, r}(t, \omega) &= \hatE\Big[ C(s, y, r,  X^{t, \omega, \hat \bu}_\rdot, \hat \bu(r,  X^{t, \omega, \hat \bu}_\rdot)) \Big| \cF_t \Big], \quad 0 \leq t \leq r.
	\end{align*}
	
	By (1)-(4) of Definition \ref{Def:EP-HJB}, 
	\begin{align}\label{Eq:ReducedHJB1-4}
	& (\bA^{\hat \bu} V)(t, \omega) + C(t, \omega_t, t, \omega_{t \wedge \cdot},  \hat \bu(t, \omega_{t \wedge \cdot})) - \int^T_t (\bA^{\hat \bu} c^r)(t, \omega, t, \omega_t) dr \cr
	& - (\bA^{\hat \bu} f)(t, \omega, t, \omega_t)  - \bA^{\hat \bu}(G \diamond g)(t, \omega) = 0.
	\end{align} 
	
	As $\hat \bu$ is admissible and $V$ satisfies Assumption \ref{Assum:ValueFunc}, we apply the functional It\^o formula in Theorem \ref{Thm:Ito} to $V$ and then claim that $ \ang{\partial_\omega V, \sigma^{t, \hat \bu}} \cdot \hat{W}$ is a true martingale, where the Brownian motion $\hat{W}$ is a part of the weak solution under $\hat{\bu}$. Along with (\ref{Eq:ReducedHJB1-4}), we obtain the following:
	\begin{align}
	& \hatE \Big[ V(T, X^{t, \omega, \hat \bu}_{T \wedge \cdot})\Big| \cF_t \Big] = \hatE \Big[ V(T, X^{t, \omega, \hat \bu} \otimes_T \Theta^{T, \hat \bu})\Big| \cF_t \Big] \cr
	& = V(t, X^{t, \omega, \hat \bu} \otimes_t \Theta^{t, \hat \bu}) + \hatE \Big[ \int^T_t (\bA^{\hat \bu} V)(s, X^{t, \omega, \hat \bu} \otimes_s \Theta^{s, \hat \bu}) ds \Big| \cF_t \Big] \cr
	& =  V(t, \omega)  - \hatE \Big[ \int^T_t C(s, X^{t, \omega, \hat \bu} \otimes_s \Theta^{s, \hat \bu}, s, X^{t, \omega, \hat \bu}_{s \wedge \cdot}, \hat{\bu}( s, X^{t, \omega, \hat \bu}_{s \wedge \cdot})) ds \Big| \cF_t \Big] \\
	&\quad +  \hatE \Big[ \int^T_t \int^T_s (\bA^{\hat \bu} c^r)(s, X^{t, \omega, \hat \bu} \otimes_s \Theta^{s, \hat \bu}, s, X^{t, \omega, \hat \bu}_s ) dr ds \Big| \cF_t \Big] \cr
	&\quad +  \hatE \Big[ \int^T_t (\bA^{\hat \bu} f)(s, X^{t, \omega, \hat \bu} \otimes_s \Theta^{s, \hat \bu}, s, X^{t, \omega, \hat \bu}_s ) ds \Big| \cF_t \Big] \cr
	&\quad +  \hatE \Big[ \int^T_t \bA^{\hat \bu}(G \diamond g)(s, X^{t, \omega, \hat \bu} \otimes_s \Theta^{s, \hat \bu}) ds \Big| \cF_t \Big]. \nonumber
	\end{align}
	
	For the third term, Fubini's theorem holds under the polynomial growth rate condition on the derivatives of $c^r$ by Assumption \ref{Assum:ValueFunc} and the conditions in (1)-(2) of Definition \ref{Def:U}. Lemma \ref{Lem:TrueMartingale} shows that $ \ang{\partial_\omega c^r, \sigma^{t, \hat \bu}} \cdot \hat{W}$ is a true martingale. Hence, the definition of $c^r$ leads to
	\begin{align*}
	&\hatE \Big[ \int^T_t \int^T_s (\bA^{\hat \bu} c^r)(s, X^{t, \omega, \hat \bu} \otimes_s \Theta^{s, \hat \bu}, s, X^{t, \omega, \hat \bu}_s ) dr ds \Big| \cF_t \Big] \\
	& = \hatE \Big[ \int^T_t \int^r_t (\bA^{\hat \bu} c^r)(s, X^{t, \omega, \hat \bu} \otimes_s \Theta^{s, \hat \bu}, s, X^{t, \omega, \hat \bu}_s ) ds dr \Big| \cF_t \Big] \\
	& = \int^T_t \hatE \Big[ \int^r_t (\bA^{\hat \bu} c^r)(s, X^{t, \omega, \hat \bu} \otimes_s \Theta^{s, \hat \bu}, s, X^{t, \omega, \hat \bu}_s ) ds  \Big| \cF_t \Big] dr\\
	& = \int^T_t \Big\{ \hatE \Big[ c^r(r, X^{t, \omega, \hat \bu} \otimes_r \Theta^{r, \hat \bu}, r, X^{t, \omega, \hat \bu}_r )\Big| \cF_t \Big] - c^r(t, X^{t, \omega, \hat \bu} \otimes_t \Theta^{t, \hat \bu}, t, X^{t, \omega, \hat \bu}_t ) \Big\}dr\\
	&= \int^T_t \Big\{ \hatE\Big[ C(r, X^{t, \omega, \hat \bu}_r, r,  X^{t, \omega, \hat \bu}_\rdot, \hat \bu(r,  X^{t, \omega, \hat \bu}_\rdot)) \Big| \cF_t \Big] \\
	&\qquad \qquad -  \hatE\Big[ C(t, X^{t, \omega, \hat \bu}_t, r,  X^{t, \omega, \hat \bu}_\rdot, \hat \bu(r,  X^{t, \omega, \hat \bu}_\rdot)) \Big| \cF_t \Big] \Big\} dr.
	\end{align*}
	
	For the fourth term, we use the same arguments:
	\begin{align*}
	& \hatE \Big[ \int^T_t (\bA^{\hat \bu} f)(s, X^{t, \omega, \hat \bu} \otimes_s \Theta^{s, \hat \bu}, s, X^{t, \omega, \hat \bu}_s ) ds \Big| \cF_t \Big] \\
	& = \hatE \Big[ f (T, X^{t, \omega, \hat \bu} \otimes_T \Theta^{T, \hat \bu}, T, X^{t, \omega, \hat \bu}_T )  \Big| \cF_t \Big] -  f (t, X^{t, \omega, \hat \bu} \otimes_t \Theta^{t, \hat \bu}, t, X^{t, \omega, \hat \bu}_t) \\
	& = \hatE \Big[ F (T, X^{t, \omega, \hat \bu}_T, X^{t, \omega, \hat \bu}_{T \wedge \cdot})  \Big| \cF_t \Big] - \hatE\Big[ F(t, \omega_t, X^{t, \omega, \hat \bu}_{T \wedge \cdot}) \Big| \cF_t\Big].
	\end{align*}
	
	Similarly, for the fifth term,
	\begin{align*}
	& \hatE \Big[ \int^T_t \bA^{\hat \bu}(G \diamond g)(s, X^{t, \omega, \hat \bu} \otimes_s \Theta^{s, \hat \bu}) ds \Big| \cF_t \Big] \\
	& = \hatE \Big[ (G\diamond g) (T, X^{t, \omega, \hat \bu} \otimes_T \Theta^{T, \hat \bu})  \Big| \cF_t \Big] -  (G\diamond g) (t, X^{t, \omega, \hat \bu} \otimes_t \Theta^{t, \hat \bu}) \\
	& = \hatE \Big[ G(T, X^{t, \omega, \hat \bu}_T, X^{t, \omega, \hat \bu}_T)  \Big| \cF_t \Big] - G(t, \omega_t, \hatE[ X^{t, \omega, \hat \bu}_T | \cF_t]).
	\end{align*}
	
	By the boundary condition in \eqref{Eq:V}, we obtain the following:
	\begin{align*}
	V(t, \omega) = &  \hatE \Big[ V(T, X^{t, \omega, \hat \bu}_{T \wedge \cdot})\Big| \cF_t \Big] + \int^T_t \hatE\Big[ C(t, X^{t, \omega, \hat \bu}_t, r,  X^{t, \omega, \hat \bu}_\rdot, \hat \bu(r,  X^{t, \omega, \hat \bu}_\rdot)) \Big| \cF_t \Big] dr \\
	& - \hatE \Big[ F (T, X^{t, \omega, \hat \bu}_T, X^{t, \omega, \hat \bu}_{T \wedge \cdot})  \Big| \cF_t \Big] + \hatE\Big[ F(t, \omega_t, X^{t, \omega, \hat \bu}_{T \wedge \cdot}) \Big| \cF_t\Big] \\
	&- \hatE \Big[ G (T, X^{t, \omega, \hat \bu}_T, X^{t, \omega, \hat \bu}_T)  \Big| \cF_t \Big] + G(t, \omega_t, \hatE[ X^{t, \omega, \hat \bu}_T | \cF_t]) \\
	=&  \int^T_t \hatE\Big[ C(t, X^{t, \omega, \hat \bu}_t, r,  X^{t, \omega, \hat \bu}_\rdot, \hat \bu(r,  X^{t, \omega, \hat \bu}_\rdot)) \Big| \cF_t \Big] dr \\
	& + \hatE\Big[ F(t, \omega_t, X^{t, \omega, \hat \bu}_{T \wedge \cdot}) \Big| \cF_t\Big] + G(t, \omega_t, \hatE[ X^{t, \omega, \hat \bu}_T | \cF_t]) \\
	=& J(t, \omega; \hat \bu).
	\end{align*}
	In other words, we verify that $V$ is the value function with $\hat \bu$.
	
	Next, we show that $\hat \bu$ is indeed an equilibrium strategy under Definition \ref{Def:Equilibrium}. We apply the recursive relationship in Lemma \ref{Lem:Recursion} with $\bu_h$. Note that $\omega^{t+h} =X^{t, \omega, \bu_h} \otimes_{t+h} \Theta^{t + h, \bu_h}$, then
	\begin{align}\label{Eq:ReCost1}
	J(t, \omega; \bu_h) & =  \hE \Big[ J(t+h, X^{t, \omega, \bu_h} \otimes_{t+h} \Theta^{t + h, \bu_h}; \bu_h) \Big| \cF_t \Big]\\
	&\quad - \Big\{ \int^T_{t+h} \hE \Big[ c^{r, \bu_h}(t+h, X^{t, \omega, \bu_h} \otimes_{t+h} \Theta^{t + h, \bu_h}, t+h, X^{t, \omega, \bu_h}_{t+h}) \Big| \cF_t \Big] dr \cr
	&\qquad \;\;- \int^T_t \hE\Big[ c^{r, \bu_h}(t+h, X^{t, \omega, \bu_h} \otimes_{t+h} \Theta^{t + h, \bu_h}, t, \omega_t)\Big| \cF_t \Big] dr \Big\} \cr
	&\quad - \Big\{ \hE \Big[ f^{\bu_h}(t+h, X^{t, \omega, \bu_h} \otimes_{t+h} \Theta^{t + h, \bu_h}, t+h, X^{t, \omega, \bu_h}_{t+h}) \Big| \cF_t \Big] \cr
	&\qquad \;\; - \hE\Big[ f^{\bu_h}(t+h, X^{t, \omega, \bu_h} \otimes_{t+h} \Theta^{t + h, \bu_h}, t, \omega_t) \Big| \cF_t \Big] \Big\}\cr
	&\quad -\Big\{ \hE \Big[ G(t+h, X^{t, \omega, \bu_h}_{t+h}, g^{\bu_h}(t+h, X^{t, \omega, \bu_h} \otimes_{t+h} \Theta^{t + h, \bu_h})) \Big| \cF_t \Big]\cr
	&\qquad \;\; - G\big(t, \omega_t, \hE\big[ g^{\bu_h}(t+h, X^{t, \omega, \bu_h} \otimes_{t+h} \Theta^{t + h, \bu_h}) \big| \cF_t\big] \big) \Big\}. \nonumber
	\end{align}
	As $\bu_h = \hat \bu $ on $[t+h, T]$, conditional on $\cF_{t+h}$, $(X^{t, \omega, \bu_h}_s)_{s \in [t+h, T]}$ has the same distribution as $(X^{t, \omega, \hat{\bu}}_s)_{s \in [t+h, T]}$. Then,
	\begin{equation}
	J(t+h, X^{t, \omega, \bu_h} \otimes_{t+h} \Theta^{t + h, \bu_h}; \bu_h) = V(t+h, X^{t, \omega, \bu_h} \otimes_{t+h} \Theta^{t + h, \bu_h}).
	\end{equation}
	For $t+h \leq r \leq T$,
	\begin{align}
	& c^{r, \bu_h}(t+h, X^{t, \omega, \bu_h} \otimes_{t+h} \Theta^{t + h, \bu_h}, t+h, X^{t, \omega, \bu_h}_{t+h}) \cr
	= & \hE \Big[ C(t+h, X^{t, \omega, \bu_h}_{t+h}, r, X^{t, \omega, \bu_h}_\rdot, \hat \bu( r, X^{t, \omega, \bu_h}_\rdot)) \Big| \cF_{t+h} \Big] \\
	= & c^r(t+h, X^{t, \omega, \bu_h} \otimes_{t+h} \Theta^{t + h, \bu_h}, t+h, X^{t, \omega, \bu_h}_{t+h}). \nonumber
	\end{align}
	
	When $t \leq r \leq t+h$,
	\begin{equation}
	c^{r, \bu_h}(t+h, X^{t, \omega, \bu_h} \otimes_{t+h} \Theta^{t + h, \bu_h}, t, \omega_t) = C(t, \omega_t, r, X^{t, \omega, \bu_h}_\rdot, \bu( r, X^{t, \omega, \bu_h}_\rdot)).
	\end{equation}
	
	When $t+h \leq r \leq T$,
	\begin{align}
	& c^{r, \bu_h}(t+h, X^{t, \omega, \bu_h} \otimes_{t+h} \Theta^{t + h, \bu_h}, t, \omega_t) \cr
	& = \hE \Big[ C(t, \omega_t, r, X^{t, \omega, \bu_h}_\rdot, \hat \bu( r, X^{t, \omega, \bu_h}_\rdot)) \Big| \cF_{t+h} \Big] \cr
	& = c^r(t+h, X^{t, \omega, \bu_h} \otimes_{t+h} \Theta^{t + h, \bu_h}, t, \omega_t). 
	\end{align}
	
	Similarly,
	\begin{align}
	& f^{\bu_h}(t+h, X^{t, \omega, \bu_h} \otimes_{t+h} \Theta^{t + h, \bu_h}, t+h, X^{t, \omega, \bu_h}_{t+h}) \\
	& = \hE \Big[ F(t+h, X^{t, \omega, \bu_h}_{t+h}, X^{t, \omega, \bu_h}_{T \wedge \cdot}) \Big| \cF_{t+h} \Big] \cr
	& =  f(t+h, X^{t, \omega, \bu_h} \otimes_{t+h} \Theta^{t + h, \bu_h}, t+h, X^{t, \omega, \bu_h}_{t+h}), \nonumber
	\end{align}
	\begin{align}
	& f^{\bu_h}(t+h, X^{t, \omega, \bu_h} \otimes_{t+h} \Theta^{t + h, \bu_h}, t, \omega_t) =  f(t+h, X^{t, \omega, \bu_h} \otimes_{t+h} \Theta^{t + h, \bu_h}, t, \omega_t),
	\end{align}
	and
	\begin{align}
	g^{\bu_h}(t+h, X^{t, \omega, \bu_h} \otimes_{t+h} \Theta^{t + h, \bu_h}) & = \hE \big[X^{t, \omega, \bu_h}_T \big| \cF_{t+h} \big] \\
	& =  g(t+h, X^{t, \omega, \bu_h} \otimes_{t+h} \Theta^{t + h, \bu_h}). \nonumber
	\end{align}
	
	Therefore, (\ref{Eq:ReCost1}) is reduced to
	\begin{align}\label{Eq:ReCost2}
	J(t, \omega; \bu_h) & =  \hE \Big[ V(t+h, X^{t, \omega, \bu_h} \otimes_{t+h} \Theta^{t + h, \bu_h}) \Big| \cF_t \Big] \\
	&\quad - \Big\{ \int^T_{t+h} \hE \big[ c^r(t+h, X^{t, \omega, \bu_h} \otimes_{t+h} \Theta^{t + h, \bu_h}, t+h, X^{t, \omega, \bu_h}_{t+h}) \big| \cF_t \big] dr \cr
	&\qquad \;\;- \int^{t+h}_t \hE\big[  C(t, \omega_t, r, X^{t, \omega, \bu_h}_\rdot, \bu( r, X^{t, \omega, \bu_h}_\rdot)) \big| \cF_t \big] dr \cr
	&\qquad \;\;- \int^T_{t+h} \hE\big[ c^r(t+h, X^{t, \omega, \bu_h} \otimes_{t+h} \Theta^{t + h, \bu_h}, t, \omega_t) \big| \cF_t \big] dr \Big\} \cr
	&\quad - \Big\{ \hE \big[ f(t+h, X^{t, \omega, \bu_h} \otimes_{t+h} \Theta^{t + h, \bu_h}, t+h, X^{t, \omega, \bu_h}_{t+h}) \big| \cF_t \big] \cr
	&\qquad \;\; - \hE\big[  f(t+h, X^{t, \omega, \bu_h} \otimes_{t+h} \Theta^{t + h, \bu_h}, t, \omega_t) \big| \cF_t \big] \Big\}  \cr
	&\quad -\Big\{ \hE \big[ G(t+h, X^{t, \omega, \bu_h}_{t+h}, g(t+h, X^{t, \omega, \bu_h} \otimes_{t+h} \Theta^{t + h, \bu_h})) \big| \cF_t \big] \cr
	&\qquad \;\; - G\big(t, \omega_t, \hE\big[ g(t+h, X^{t, \omega, \bu_h} \otimes_{t+h} \Theta^{t + h, \bu_h}) \big| \cF_t\big] \big) \Big\}. \nonumber
	\end{align}
	
	Meanwhile, from the PHJB equation \eqref{Eq:V} for $V$, we apply the functional It\^o formula and Lemma \ref{Lem:TrueMartingale}. Note the right-continuity in time and continuity in $\omega$ assumption from (2)-(3) of Definition \ref{Def:U}:
	\begin{align}\label{Eq:ItoCost1}
	& \hE \Big[ V(t+h, X^{t, \omega, \bu_h} \otimes_{t+h} \Theta^{t + h, \bu_h}) \Big| \cF_t \Big] - V(t, \omega) \cr
	& + \hE \Big[ \int^{t+h}_t C(s, X^{t, \omega, \bu_h}_s, s, X^{t, \omega, \bu_h}_{s \wedge \cdot}, \bu( s, X^{t, \omega, \bu_h}_{s \wedge \cdot}) ) ds \Big| \cF_t \Big] \cr
	& - \hE \Big[ \int^{t+h}_t \int^T_s (\bA^{\bu_h} c^r)(s, X^{t, \omega, \bu_h} \otimes_s \Theta^{s, \bu_h}, s, X^{t, \omega, \bu_h}_s) dr ds \Big| \cF_t \Big] \cr 
	& + \hE \Big[ \int^{t+h}_t \int^T_s (\bA^{\bu_h} c^{t, \omega_t, r})(s, X^{t, \omega, \bu_h} \otimes_s \Theta^{s, \bu_h}) dr ds \Big| \cF_t \Big] \\
	& - \Big\{ \hE \Big[ f(t+h, X^{t, \omega, \bu_h} \otimes_{t+h} \Theta^{t + h, \bu_h}, t+h, X^{t, \omega, \bu_h}_{t+h}) \Big| \cF_t \Big] - f(t, \omega, t, \omega_t) \Big\} \cr
	& + \Big\{ \hE \Big[ f(t+h, X^{t, \omega, \bu_h} \otimes_{t+h} \Theta^{t + h, \bu_h}, t, \omega_t) \Big| \cF_t \Big] - f(t, \omega, t, \omega_t) \Big\} \cr 
	& - \Big\{ \hE \Big[ G(t+h, X^{t, \omega, \bu_h}_{t+h}, g(t+h, X^{t, \omega, \bu_h} \otimes_{t+h} \Theta^{t + h, \bu_h})) \Big| \cF_t \Big] - G(t, \omega_t, g(t, \omega)) \Big\} \cr
	& + \Big\{ G\Big(t, \omega_t, \hE \big[ g(t+h, X^{t, \omega, \bu_h} \otimes_{t+h} \Theta^{t + h, \bu_h}) \big| \cF_t \big] \Big) - G(t, \omega_t, g(t, \omega)) \Big\} \leq o(h). \nonumber
	\end{align}
	
	We further simplify the $C, c^r$, and $c^{t, \omega_t, r}$ terms in (\ref{Eq:ItoCost1}). By Fubini's theorem, we obtain the following:
	\begin{align}
	&\hE \Big[\int^{t+h}_t \int^T_s (\bA^{\bu_h} c^r)(s, X^{t, \omega, \bu_h} \otimes_s \Theta^{s, \bu_h}, s, X^{t, \omega, \bu_h}_s) dr ds \Big| \cF_t \Big] \cr
	&= \hE \Big[ \int^T_t \int^{r \wedge (t+h)}_t  (\bA^{\bu_h} c^r)(s, X^{t, \omega, \bu_h} \otimes_s \Theta^{s, \bu_h}, s, X^{t, \omega, \bu_h}_s) ds dr \Big| \cF_t \Big] \cr
	&= \int^{t+h}_t \hE \Big[ \int^r_t (\bA^{\bu_h} c^r)(s, X^{t, \omega, \bu_h} \otimes_s \Theta^{s, \bu_h}, s, X^{t, \omega, \bu_h}_s) ds \Big| \cF_t \Big] dr \label{Eq:crFubini} \\
	&\quad +  \int^T_{t+h}\hE \Big[ \int^{t+h}_t (\bA^{\bu_h} c^r)(s, X^{t, \omega, \bu_h} \otimes_s \Theta^{s, \bu_h}, s, X^{t, \omega, \bu_h}_s) ds \Big| \cF_t \Big] dr \cr
	&=\int^{t+h}_t \Big\{ \hE \big[ c^r(r, X^{t, \omega, \bu_h} \otimes_r \Theta^{r, \bu_h}, r, X^{t, \omega, \bu_h}_r) \big| \cF_t \big] - c^r(t, \omega, t, \omega_t) \Big\}dr \cr
	&\quad + \int^T_{t+h} \Big\{ \hE \big[ c^r(t+h, X^{t, \omega, \bu_h} \otimes_{t+h} \Theta^{t+h, \bu_h}, t+h, X^{t, \omega, \bu_h}_{t+h}) \big| \cF_t \big] - c^r(t, \omega, t, \omega_t) \Big\}dr. \nonumber
	\end{align}
	Similarly,
	\begin{align}
	&\hE \Big[ \int^{t+h}_t \int^T_s (\bA^{\bu_h} c^{t, \omega_t, r})(s, X^{t, \omega, \bu_h} \otimes_s \Theta^{s, \bu_h}) dr ds \Big| \cF_t \Big] \cr
	&=\int^{t+h}_t \Big\{ \hE \big[ c^r(r, X^{t, \omega, \bu_h} \otimes_r \Theta^{r, \bu_h}, t, \omega_t) \big| \cF_t \big] - c^r(t, \omega, t, \omega_t) \Big\}dr \label{Eq:ctrFubini} \\
	&\quad + \int^T_{t+h} \Big\{ \hE \big[ c^r(t+h, X^{t, \omega, \bu_h} \otimes_{t+h} \Theta^{t+h, \bu_h}, t, \omega_t) \big| \cF_t \big] - c^r(t, \omega, t, \omega_t) \Big\}dr \nonumber
	\end{align}
	and
	\begin{align}
	& \int^{t+h}_t \hE\big[ c^r(r, X^{t, \omega, \bu_h} \otimes_r \Theta^{r, \bu_h}, r, X^{t, \omega, \bu_h}_r) \big| \cF_t \big]dr \cr
	& = \int^{t+h}_t \hE\big[ \hE\big[ C(r, X^{t, \omega, \bu_h}_r, r, X^{t, \omega, \bu_h}_\rdot, \hat \bu(r, X^{t, \omega, \bu_h}_\rdot)) \big| \cF_r \big] \big| \cF_t \big]dr \\
	& = \int^{t+h}_t \hE\big[ C(r, X^{t, \omega, \bu_h}_r, r, X^{t, \omega, \bu_h}_\rdot, \hat \bu(r, X^{t, \omega, \bu_h}_\rdot)) \big| \cF_t \big]dr. \nonumber
	\end{align}
	\begin{align}
	& \int^{t+h}_t \hE\big[ c^r(r, X^{t, \omega, \bu_h} \otimes_r \Theta^{r, \bu_h}, t, \omega_t) \big| \cF_t \big]dr \nonumber \\
	& = \int^{t+h}_t \hE\big[ C(t, \omega_t, r, X^{t, \omega, \bu_h}_\rdot, \hat \bu(r, X^{t, \omega, \bu_h}_\rdot)) \big| \cF_t \big]dr.
	\end{align}
	
	By the Lebesgue differentiation theorem held under (3)-(4) of Definition \ref{Def:U}, we obtain the following:
	\begin{align}\label{Eq:CDiffepsilon}
	& \int^{t+h}_t \hE\big[ C(r, X^{t, \omega, \bu_h}_r, r, X^{t, \omega, \bu_h}_\rdot, \hat \bu(r, X^{t, \omega, \bu_h}_\rdot)) \big| \cF_t \big]dr \cr
	& = \int^{t+h}_t \hE\big[ C(t, \omega_t, r, X^{t, \omega, \bu_h}_\rdot, \hat \bu(r, X^{t, \omega, \bu_h}_\rdot)) \big| \cF_t \big]dr + o(h).
	\end{align}
	
	Thus, (\ref{Eq:crFubini}) and (\ref{Eq:ctrFubini}) are reduced to
	\begin{align}\label{Eq:ReduceC_T}
	& - \hE \Big[ \int^{t+h}_t \int^T_s (\bA^{\bu_h} c^r)(s, X^{t, \omega, \bu_h} \otimes_s \Theta^{s, \bu_h}, s, X^{t, \omega, \bu_h}_s) dr ds \Big| \cF_t \Big] \cr 
	& + \hE \Big[ \int^{t+h}_t \int^T_s (\bA^{\bu_h} c^{t, \omega_t, r})(s, X^{t, \omega, \bu_h} \otimes_s \Theta^{s, \bu_h}) dr ds \Big| \cF_t \Big] \cr
	= & -\int^T_{t+h} \hE \big[ c^r(t+h, X^{t, \omega, \bu_h} \otimes_{t+h} \Theta^{t+h, \bu_h}, t+h, X^{t, \omega, \bu_h}_{t+h}) \big| \cF_t \big]dr \\
	& + \int^T_{t+h} \hE \big[ c^r(t+h, X^{t, \omega, \bu_h} \otimes_{t+h} \Theta^{t+h, \bu_h}, t, \omega_t) \big| \cF_t \big] dr + o(h). \nonumber
	\end{align}
	Using the same argument in (\ref{Eq:CDiffepsilon}) yields the following:
	\begin{align}\label{Eq:ReduceC_th}
	& \hE \Big[ \int^{t+h}_t C(s, X^{t, \omega, \bu_h}_s, s, X^{t, \omega, \bu_h}_{s \wedge \cdot}, \bu( s, X^{t, \omega, \bu_h}_{s \wedge \cdot}) ) ds \Big| \cF_t \Big] \\
	& = \hE\big[ \int^{t+h}_t  C(t, \omega_t, r, X^{t, \omega, \bu_h}_\rdot, \bu( r, X^{t, \omega, \bu_h}_\rdot)) dr \big| \cF_t \big] + o(h). \nonumber
	\end{align}
	
	Combining (\ref{Eq:ItoCost1}), (\ref{Eq:ReduceC_T}), and (\ref{Eq:ReduceC_th}) yields
	\begin{align}\label{Eq:ItoCost2}
	& \hE \Big[ V(t+h, X^{t, \omega, \bu_h} \otimes_{t+h} \Theta^{t + h, \bu_h}) \Big| \cF_t \Big] - V(t, \omega) \cr
	& + \hE\big[ \int^{t+h}_t  C(t, \omega_t, r, X^{t, \omega, \bu_h}_\rdot, \bu( r, X^{t, \omega, \bu_h}_\rdot)) dr \big| \cF_t \big] \cr
	&-\int^T_{t+h} \hE \big[ c^r(t+h, X^{t, \omega, \bu_h} \otimes_{t+h} \Theta^{t+h, \bu_h}, t+h, X^{t, \omega, \bu_h}_{t+h}) \big| \cF_t \big]dr \cr
	& + \int^T_{t+h} \hE \big[ c^r(t+h, X^{t, \omega, \bu_h} \otimes_{t+h} \Theta^{t+h, \bu_h}, t, \omega_t) \big| \cF_t \big] dr \\
	& - \hE \Big[ f(t+h, X^{t, \omega, \bu_h} \otimes_{t+h} \Theta^{t + h, \bu_h}, t+h, X^{t, \omega, \bu_h}_{t+h}) \Big| \cF_t \Big] \cr
	& + \hE \Big[ f(t+h, X^{t, \omega, \bu_h} \otimes_{t+h} \Theta^{t + h, \bu_h}, t, \omega_t) \Big| \cF_t \Big] \cr 
	& - \hE \Big[ G(t+h, X^{t, \omega, \bu_h}_{t+h}, g(t+h, X^{t, \omega, \bu_h} \otimes_{t+h} \Theta^{t + h, \bu_h})) \Big| \cF_t \Big] \cr
	& + G\Big(t, \omega_t, \hE \big[ g(t+h, X^{t, \omega, \bu_h} \otimes_{t+h} \Theta^{t + h, \bu_h}) \big| \cF_t \big] \Big) \leq o(h). \nonumber
	\end{align}
	
	Compared with (\ref{Eq:ReCost2}), we ensure that
	\begin{equation}
	J(t, \omega; \bu_h) - V(t, \omega) \leq o(h).
	\end{equation}
	As $V(t, \omega) = J(t, \omega; \hat \bu)$ is shown previously,
	\begin{equation}
	J(t, \omega; \bu_h) - J(t, \omega; \hat \bu) \leq o(h),
	\end{equation}
	as desired.
\end{proof}

\subsubsection{Proof of Corollary \ref{Cor:Logpsi}}
\begin{proof}
	We first prove the results for $\psi$. Consider $ \tilde \psi = - \psi$. Then, $\tilde \psi $ satisfies
	\begin{equation}\label{Eq:tildepsi}
	\tilde \psi = K * H(\psi).
	\end{equation} 
	$w_*$ is the unique root of $H(w) = 0$ on $(-\infty, w_{max}]$ with $ w_{max}  \triangleq - \frac{H_1}{2H_2}$. $H(w)$ satisfies Assumption A.1 in \cite{gatheral2019affine}. Therefore, \citet[Theorem A.5 (a)]{gatheral2019affine} with $a(t) \equiv 0$ implies that (\ref{Eq:tildepsi}) has a unique global continuous solution and that
	\begin{equation}
	w_* < r_1(t) \leq \tilde \psi(t) < 0, \quad \forall \; t > 0.
	\end{equation}
	This gives the desired result for $\psi$.

	\eqref{Eq:LogMV_V2} is a linear VIE. The existence and uniqueness results of $V_2$ are thus determined using \citet[Theorem 1.2.3]{brunner2017volterra} or \citet[Equation (1.3), p.77]{gripenberg1990volterra}. \eqref{Eq:LogMV_V0} and \eqref{Eq:LogMVg0} are linear ordinary differential equations (ODEs).
\end{proof}

\subsubsection{Proof of Corollary \ref{Cor:LogMoment}}
\begin{proof}
	$M^*$ under $\hat{\bpi}$ is given by
	\begin{equation}
	M^*_t = M_0 \exp\Big[ \int^t_0 (\varUpsilon_s + \theta \hat{\pi}_s \nu_s - \frac{1}{2} \hat{\pi}^2_s \nu_s)ds + \int^t_0 \sqrt{\nu_s} \hat{\pi}_s d \hat{W}_{1s} \Big],
	\end{equation}
	where Brownian motion $\hat{W}_1$ is a part of the weak solution under $\hat{\bpi}$.
 	
	By Doob's maximal inequality and \citet[Lemma 7.3]{abi2017affine},
	\begin{align*}
	& \hatE \Big[ \sup_{ t \in [0, T]} | M^*_t |^p \Big] \\
	& \leq C \hatE \Big[ \sup_{ t \in [0, T]} \Big| e^{- \int^t_0 \theta \hat{\pi}_s \nu_s ds} \Big|^{2p} \Big] + C \hatE \Big[ \sup_{ t \in [0, T]} \Big|  \exp \Big(- \frac{1}{2} \int^t_0  \hat{\pi}^2_s \nu_s ds + \int^t_0 \sqrt{\nu_s} \hat{\pi}_s d \hat{W}_{1s}  \Big) \Big|^{2p}\Big] \\
	& \leq C \hatE \Big[ e^{2p \int^T_0 |\theta \hat{\pi}_s| \nu_s ds} \Big] + C \hatE \Big[ \exp \Big(- \int^T_0 p \hat{\pi}^2_s  \nu_s ds + \int^T_0 2p \hat{\pi}_s \sqrt{\nu_s} d \hat{W}_{1s} \Big) \Big].
	\end{align*}
	The first term is finite by assumption \eqref{Assum:LogExpMon} with constant $2 p |\theta| \sup_{t \in [0, T]} | \hat{\pi}_t |$. The second term is also finite. In fact, by H\"older's inequality and assumption \eqref{Assum:LogExpMon} with a constant $(8p^2 - 2p) \sup_{t \in [0, T]}  \hat{\pi}^2_t$,
	\begin{align*}
	& \hatE \Big[ \exp \Big(- \int^T_0 p \hat{\pi}^2_s  \nu_s ds + \int^T_0 2p \hat{\pi}_s \sqrt{\nu_s} d \hat{W}_{1s}  \Big) \Big] \\
	& \leq \Big\{ \hatE \Big[ e^{(8p^2 - 2p) \int^T_0 \hat{\pi}^2_s  \nu_s ds} \Big] \Big\}^{1/2} \Big\{ \hatE \Big[ \exp \Big(- 8 p^2 \int^T_0 \hat{\pi}^2_s  \nu_s ds + 4 p \int^T_0 \hat{\pi}_s \sqrt{\nu_s} d \hat{W}_{1s}  \Big) \Big] \Big\}^{1/2} \\
	& < \infty.
	\end{align*}
\end{proof}

\subsubsection{Proof of Corollary \ref{Cor:F}}
\begin{proof}
	When $\lambda = 0$, the claim is directly verifiable. We suppose that $\lambda \neq 0$. By \citet[Equation (3.4)]{mainardi2014some} or \citet[Appendix A.1]{el2019characteristic},
	\begin{align}
	\frac{F^{\alpha, \lambda}(T - t)}{\lambda} \underset{T-t \rightarrow \infty}{\sim} \frac{1}{\lambda} - \frac{1}{\lambda^2 \Gamma(1 - \alpha) (T - t)^\alpha}.
	\end{align}
	Note that
	\begin{align}
	\partial_\alpha \big[\Gamma(1 - \alpha) (T - t)^\alpha\big] = \big[\ln(T - t) - \psi( 1 - \alpha)\big] \Gamma(1 - \alpha) (T - t)^\alpha > 0,
	\end{align}
	where $\psi(z) = \frac{\Gamma'(z)}{\Gamma(z)}$ is the polygamma function and $\psi( 1 - \alpha) < 0$. We get the first part of the claim.
	
	Furthermore,
	\begin{align}
	\frac{F^{\alpha, \lambda}(T - t)}{\lambda} \underset{T-t \rightarrow 0^+}{\sim} \frac{(T - t)^\alpha}{\Gamma(\alpha + 1)}.
	\end{align}
	When $T - t$ is small, $(T - t)^\alpha$ is decreasing on $\alpha$. Finally, $\frac{(T - t)^\alpha}{\Gamma(\alpha + 1)}$ decreases with $\alpha$.
\end{proof}


\end{document}